\RequirePackage[l2tabu,orthodox]{nag}		
\documentclass[reqno]{amsart}		
\usepackage[margin=1.5in,bottom=1.25in]{geometry}		

\usepackage{amsmath}		
\usepackage{amssymb}		
\usepackage{amsfonts}		
\usepackage{amsthm}		
\usepackage[foot]{amsaddr}		

\usepackage{mathtools}		
\mathtoolsset{%
}

\newtagform{comment}{\{}{\}}

\usepackage{thmtools}		
\usepackage{thm-restate}		

\usepackage[utf8]{inputenc}		
\usepackage[T1]{fontenc}		

\usepackage[proportional,tabular,lining,sf,mono=false]{libertine}

\usepackage{dsfont}		

\usepackage{inconsolata}

\usepackage[
cal=cm,
]
{mathalfa}

\usepackage[labelfont={bf,small},labelsep=colon,font=small]{caption}	
\usepackage[dvipsnames,svgnames]{xcolor}		
\colorlet{MyBlue}{Navy}
\colorlet{MyGreen}{DarkGreen!85!Black}

\usepackage{setspace}

\newcommand{\afterhead}{.}		
\newcommand{\para}[1]{\medskip\paragraph{\textbf{#1\afterhead}}}

\usepackage{wasysym}		

\usepackage{tikz}		
\usetikzlibrary{calc,patterns}

\usepackage{array}		
\usepackage{booktabs}		
\usepackage{paralist}		

\usepackage{microtype}		

\usepackage{acronym}		
\usepackage{latexsym}		
\usepackage{xspace}		

\usepackage[numbers,sort&compress]{natbib}		

\bibpunct[, ]{(}{)}{,}{}{,}{,}

\usepackage{hyperref}
\hypersetup{
colorlinks=true,
linktocpage=true,
pdfstartview=FitH,
breaklinks=true,
pdfpagemode=UseNone,
pageanchor=true,
pdfpagemode=UseOutlines,
plainpages=false,
bookmarksnumbered,
bookmarksopen=false,
bookmarksopenlevel=1,
hypertexnames=true,
pdfhighlight=/O,
urlcolor=Maroon,linkcolor=MyBlue!60!black,citecolor=MyBlue!70!black,	
pdftitle={},
pdfauthor={},
pdfsubject={},
pdfkeywords={},
pdfcreator={pdfLaTeX},
pdfproducer={LaTeX with hyperref}
}

\usepackage[sort&compress,capitalize,nameinlink]{cleveref}		
\crefname{assumption}{Assumption}{Assumptions}

\usepackage[textwidth=30mm]{todonotes}		

\newcommand{\debug}[1]{#1}		

\usepackage{algorithm}		
\usepackage{algpseudocode}		
\renewcommand{\algorithmiccomment}[1]{\hfill\textsf{\#} #1}		

\theoremstyle{plain}
\newtheorem{corollary}{Corollary}		
\newtheorem{lemma}{Lemma}		
\newtheorem{proposition}{Proposition}		

\newtheorem*{corollary*}{Corollary}		

\theoremstyle{definition}
\newtheorem{assumption}{Assumption}		
\newtheorem{example}{Example}		

\newtheorem*{definition*}{Definition}		
\newtheorem*{assumption*}{Assumptions}		
\newtheorem*{example*}{Example}		

\theoremstyle{remark}

\newtheorem*{remark*}{Remark}		

\newcounter{proofpart}

\numberwithin{example}{section}		

\DeclarePairedDelimiter{\bracks}{[}{]}		
\DeclarePairedDelimiter{\parens}{(}{)}		

\DeclarePairedDelimiter{\abs}{\lvert}{\rvert}		

\DeclarePairedDelimiterX{\setdef}[2]{\{}{\}}{#1:#2}		
\DeclarePairedDelimiterXPP{\exclude}[1]{\mathopen{}\setminus}{\{}{\}}{}{#1}

\newcommand{\cf}{cf.\xspace}		
\newcommand{\eg}{e.g.,\xspace}		
\newcommand{\ie}{i.e.,\xspace}		
\newcommand{\textpar}[1]{\textup(#1\textup)}		

\newcommand{\alt}[1]{\tilde#1}		

\newcommand{\R}{\mathbb{R}}		

\DeclareMathOperator{\bigoh}{\mathcal O}		

\newcommand{\dd}{\:d}		

\newcommand{\vecspace}{\R^{\vdim}}		
\newcommand{\vdim}{\debug n}		

\newcommand{\pertvec}{\debug w}		
\newcommand{\unitvec}{\debug z}		

\DeclarePairedDelimiterX{\braket}[2]{\langle}{\rangle}{#1,#2}		

\newcommand{\cvx}{\mathcal{C}}		

\DeclareMathOperator*{\argmin}{arg\,min}		

\newcommand{\points}{\mathcal{\debug X}}		
\newcommand{\point}{\debug x}		
\newcommand{\pointalt}{\alt\point}		

\newcommand{\dpoint}{\debug y}		

\newcommand{\vbound}{\debug V_{\ast}}		

\newcommand{\smooth}{\debug \beta}		

\newcommand{\sol}[1][\point]{#1_{\ast}}		

\DeclareMathOperator{\reg}{Reg}		

\newcommand{\play}{\debug i}		
\newcommand{\playalt}{\debug j}		
\newcommand{\nPlayers}{\debug N}		
\newcommand{\players}{\mathcal{\nPlayers}}		

\newcommand{\pay}{\debug u}		
\newcommand{\payv}{\debug V}		

\newcommand{\eq}{\sol}		
\newcommand{\game}{\mathcal{\debug G}}		

\DeclareMathOperator{\Eucl}{\Pi}		

\newcommand{\breg}{\debug D}		
\newcommand{\radius}{\debug r}		

\DeclareMathOperator{\ex}{\mathbb{E}}		
\DeclareMathOperator{\prob}{\mathbb{P}}		

\newcommand{\samples}{\debug \Omega}		
\newcommand{\filter}{\mathcal{\debug F}}		

\newcommand{\as}{\textpar{a.s.}\xspace}		

\providecommand\given{}		

\DeclarePairedDelimiterXPP{\exof}[1]{\ex}{[}{]}{}{
\renewcommand\given{\nonscript\,\delimsize\vert\nonscript\,\mathopen{}} #1}

\DeclarePairedDelimiterXPP{\probof}[1]{\prob}{(}{)}{}{
\renewcommand\given{\nonscript\,\delimsize\vert\nonscript\,\mathopen{}} #1}

\newcommand{\est}[1]{\hat #1}		

\newcommand{\bias}{\debug b}		
\newcommand{\sbias}{\debug r}		
\newcommand{\bbound}{\debug R}		

\newcommand{\noise}{\debug U}		
\newcommand{\snoise}{\debug \xi}		

\newcommand{\svar}{\debug \psi}		

\newcommand{\start}{\debug 1}		
\newcommand{\running}{\debug 1,2,\dotsc}		
\newcommand{\run}{\debug t}		
\newcommand{\runalt}{\debug s}		
\newcommand{\nRuns}{\debug T}		

\newcommand{\state}{\debug X}		

\newcommand{\step}{\debug \gamma}		

\DeclareMathOperator{\one}{\mathds{1}}		

\newcommand{\from}{\colon}		

\newcommand{\base}{\debug p}		

\DeclareMathOperator{\diam}{diam}		

\DeclarePairedDelimiterX{\product}[2]{\langle}{\rangle}{#1,#2}		
\DeclarePairedDelimiter{\norm}{\lVert}{\rVert}		
\DeclarePairedDelimiterXPP{\dnorm}[1]{}{\lVert}{\rVert}{}{#1}		

\DeclareMathOperator{\vol}{vol}		

\newcommand{\ball}{\mathbb{\debug B}}		
\newcommand{\sphere}{\mathbb{\debug S}}		

\newcommand{\reward}{\debug R}		
\newcommand{\rewards}{\mathcal{\reward}}		
\newcommand{\info}{\mathcal{\debug R}}		
\newcommand{\pool}{\mathcal{\debug P}}		
\newcommand{\unused}{\mathcal{\debug U}}		
\newcommand{\head}{\debug q}		
\newcommand{\delay}{\debug d}
\newcommand{\Delay}{\debug D}

\newcommand{\mix}{\debug \delta}		
\newcommand{\unitvar}{\debug Z}		
\newcommand{\pertvar}{\debug W}		

\begin{document}


\title
[Gradient-free Online Learning in Games with Delayed Rewards]
{Gradient-free Online Learning in Games with Delayed Rewards}

\author
[A.~Héliou]
{Amélie Héliou$^{\ast}$}
\email{a.heliou@criteo.com}

\author
[P.~Mertikopoulos]
{Panayotis Mertikopoulos$^{\diamond,\ast,\lowercase{c}}$}
\email{panayotis.mertikopoulos@imag.fr}

\author
[Z.~Zhou]
{Zhengyuan Zhou$^{\sharp}$}
\email{put.email@here}

\address{$^{c}$\,Corresponding author.}
\address{$^{\ast}$\,Criteo AI Lab.}
\address{$^{\diamond}$\,Univ. Grenoble Alpes, CNRS, Inria, LIG, 38000, Grenoble, France.}
\address{$^{\sharp}$\,Stern School of Business, NYU, and IBM Research.}



\thanks{
%
%
This research was partially supported by the COST Action CA16228 ``European Network for Game Theory'' (GAMENET).
P.~Mertikopoulos is also grateful for financial support by
the French National Research Agency (ANR) under grant no.~ANR\textendash 16\textendash CE33\textendash 0004\textendash 01 (ORACLESS).
Zhengyuan Zhou is grateful for the IBM Goldstine fellowship.}

\newcommand{\acli}[1]{\textit{\acl{#1}}}		
\newcommand{\aclip}[1]{\textit{\aclp{#1}}}		
\newcommand{\acdef}[1]{\textit{\acl{#1}} \textup{(\acs{#1})}\acused{#1}}		
\newcommand{\acdefp}[1]{\emph{\aclp{#1}} \textup(\acsp{#1}\textup)\acused{#1}}	

\newacro{DSC}{diagonal strict concavity}
\newacro{SPSA}{simultaneous perturbation stochastic approximation}
\newacro{MAB}{multi-armed bandit}
\newacro{FTRL}{follow-the-regularized-leader}
\newacro{OGD}{online gradient descent}
\newacro{OMD}{online mirror descent}
\newacro{MD}{mirror descent}
\newacro{GOLD}{gradient-free online learning with delayed feedback}
\newacro{FIFO}{first-in, first-out}

\newacro{LHS}{left-hand side}
\newacro{RHS}{right-hand side}
\newacro{iid}[i.i.d.]{independent and identically distributed}
\newacro{lsc}[l.s.c.]{lower semi-continuous}
\newacro{NE}{Nash equilibrium}
\newacroplural{NE}[NE]{Nash equilibria}

\begin{abstract}
%
%
Motivated by applications to online advertising and recommender systems, we consider a game-theoretic model with delayed rewards and asynchronous, payoff-based feedback.
In contrast to previous work on delayed multi-armed bandits, we focus on multi-player games with continuous action spaces, and we examine the long-run behavior of strategic agents that follow a no-regret learning policy (but are otherwise oblivious to the game being played, the objectives of their opponents, etc.).
To account for the lack of a consistent stream of information (for instance, rewards can arrive out of order, with an a priori unbounded delay, etc.), we introduce a gradient-free learning policy where payoff information is placed in a priority queue as it arrives.
In this general context, we derive new bounds for the agents' regret;
furthermore, under a standard diagonal concavity assumption, we show that the induced sequence of play converges to \ac{NE} with probability $1$, even if the delay between choosing an action and receiving the corresponding reward is unbounded.
\end{abstract}

\maketitle
\acresetall		

\section{Introduction}
\label{sec:introduction}

A major challenge in the application of learning theory to online advertising and recommender systems is that there is often a significant delay between action and reaction:
for instance, a click on an ad can be observed within seconds of the ad being displayed, but the corresponding sale can take hours or days to occur \textendash\ if it occurs at all.
Putting aside all questions of causality and ``what if'' reasoning (\eg the attribution of the sale to a given click), this delay has an adverse effect on all levels of the characterization between marketing actions and a user's decisions.

Similar issues also arise in operations research, online machine learning, and other fields where online decision-making is the norm;
as an example, we mention here the case of traffic allocation and online path planning, signal coveriance optimization in signal processing, etc.
In view of all this, a key question that arises is
\begin{inparaenum}
[\itshape a\upshape)]
\item
to quantify the impact of a delayed reward\,/\,feedback structure on multi-agent learning;
and
\item
to design policies that exploit obsolete information in a way as to minimize said impact.
\end{inparaenum}

\para{Context}

In this paper, we examine the above questions in the general framework of online learning in games with continuous action spaces.
In more detail, we focus on recurrent decision processes that unfold as follows:
At each stage $\run=\running$,
the decision-maker (or player) selects an action $\state_{\run}$ from a set of possible actions $\points$.
This action subsequently triggers a reward $\pay_{\run}(\state_{\run})$ based on some \textpar{a priori unknown} payoff function $\pay_{\run}\from\points\to\R$.
However, in contrast to the standard online optimization setting, this reward is only received by the player $\delay_{\run}$ stages later,
\ie at round $\run + \delay_{\run}$.
As a result, the player may receive no information at round $\run$, or they may receive older, obsolete information from some previous round $\runalt < \run$.

This very broad framework places no assumptions on
the governing dynamics between actions and rewards,
the payoff-generating process,
or
the delays encountered by the player.
As such, the most common performance measure for a realized sequence of actions is the player's \emph{regret}, \ie the difference between the player's cumulative payoff over a given horizon and that of the best fixed action in hindsight.
Thus, in the absence of more refined knowledge about the environments, the most sensible choice would be to deploy a policy which, at the very least, leads to \emph{no regret}.

A specific instance of this ``agnostic'' framework \textendash\ and one that has attracted considerable interest in the literature \textendash\ is when the rewards of a given player are determined by the player's interactions with other players, even though the dynamics of these interactions can be unknown to the decision-making players beforehand.
For instance, when placing a bid for reserving ad space, the ultimate payoff of a bidder will be determined by the bids of all other participating players and the rules of the underlying auction.
The exact details of the auction (\eg its reserve price) may be unknown to the bidders, and the bidders may not know anything about whom they are bidding against, but their rewards are still determined by a fixed mechanism \textendash\ that of an \emph{$\nPlayers$-player game}.

With all this in mind, our paper focuses on the following questions that arise naturally in this context:
\emph{Is there a policy leading to no regret in online optimization problems with delayed, payoff-based feedback?}
And, assuming all players subscribe to such a policy,
\emph{does the induced sequence of play converge to a stable, equilibrium state?}

\para{Our contributions}

Our first contribution is to design a policy for online learning in this setting, which we call \acdef{GOLD}.
The backbone of this policy is the \ac{OGD} algorithm of \citet{Zin03}, but with two important modifications designed to address the challenges of the current setting.
The first modification is the inclusion of a zeroth-order gradient estimator based on the \ac{SPSA} mechanism of \citet{Spa97} and \citet{FKM05}.
By virtue of this stochastic approximation mechanism, the player can estimate \textendash\ albeit in a biased way \textendash\ the gradient of their payoff function by receiving the reward of a nearby, perturbed action.
The second element of \ac{GOLD} is the design of a novel information pooling strategy that records information in a priority queue as they arrive, and subsequently dequeues them following a \ac{FIFO} scheme.
The main challenge that occurs here is that the stream of information received by an agent may be highly unbalanced, \eg consisting of intermittent batches of obsolete information followed by periods of feedback silence.
This suggests that an agent should exercise a certain ``economy of actions'' and refrain from burning through batches of received information too quickly;
the proposed pooling policy achieves precisely this by dequeueing at most one bit of feedback, even if more is available at any given stage.

From a theoretical viewpoint, the principal difficulty that arises is how to fuse these two components and control the errors that accrue over time from the use of obsolete \textendash\ and biased \textendash\ gradient estimates.
This requires a delicate shadowing analysis and a careful tweaking of the method's parameters \textendash\ specifically, its step-size sequence and the query radius of the \ac{SPSA} estimator.
In so doing, our first theoretical result is that \ac{GOLD} guarantees no regret, even if the delays encountered by the agent are unbounded.
Specifically, if the reward of the $\run$-th round is received up to $o(\run^{\alpha})$ rounds later, then the \ac{GOLD} algorithm enjoys a regret bound of the form $\bigoh(\nRuns^{3/4} + \nRuns^{2/3 + \alpha/3})$.
In particular, this means that \ac{GOLD} guarantees no regret even under \emph{unbounded} delays that might grow over time at a sublinear rate.

Our third contribution is to derive the game-theoretic implications of concurrently running \ac{GOLD} in a multi-agent setting.
A priori, the link between no regret and \acl{NE} (as opposed to coarse correlated equilibrium) is quite weak.
Nevertheless, if the game in question satisfies a standard monotonicity condition due to \citet{Ros65}, we show that the sequence of actions generated by the \ac{GOLD} policy converges to \acl{NE} with probability $1$.
To the best of our knowledge, this is the first \acl{NE} convergence result for game-theoretic learning with delayed, payoff-based feedback.


\para{Related work}

The no-regret properties of \ac{OGD} in settings with delayed feedback was recently considered by \citet{QK15} who proposed a natural extension of \ac{OGD} where the player performs a \emph{batched} gradient update the moment gradients are received.
Doing so, \citet{QK15} showed that if the total delay over a horizon $\nRuns$ is $\Delay_{\nRuns} = \sum_{\run=\start}^{\nRuns} \delay_{\run}$, \ac{OGD} enjoys a regret bound of the form $\bigoh(\sqrt{\nRuns + \Delay_{\nRuns}})$.
This bound echoes a string of results obtained in the \ac{MAB} literature under different assumptions:
for instance,
\citet{JGS13} and \citet{VCP17} assume that the origin of the information is known;
\citet{QK15} and \citet{PBSS+18} do not make this assumption and instead consider an ``anonymized'' feedback enviroment;
etc.

When the action space is finite, online learning with delayed feedback has also been explored in the context of adversarial \acp{MAB}.
In this context, \citet{TCS19} bound the regret in this case with the cumulative delay, which, in our notation, would be $\bigoh(\nRuns^{1+\alpha})$.
Taking into account the non-square-root scaling of the regret due to the lack of gradient observations, this would conceivably lead to a bound similar to that of \cref{thm:regret} for a \ac{MAB} setting.
Related papers which provide adaptive tuning to the unknown sum of delays are the works of \citet{JGS16,ZS20}, while \citet{bistritz2019online} and \citep{zhou2019learning} provide further results in adversarial and linear contextual bandits respectively.
However, the algorithms used in these works have little to do with \ac{OGD}.

Likewise, no-regret learning in bandit convex optimization has a long history dating back at least to \citet{Kle04} and \citet{FKM05}.
The standard \ac{OGD} policy with \ac{SPSA} gradient estimates achieves an $\bigoh(\nRuns^{3/4})$ regret bound, and the $\nRuns^{3/4}$ term in our bound is indeed related to this estimate.
Using sophisticated kernel estimation techniques, \citet{BE16,BE17} decreased this bound to $\bigoh(\nRuns^{1/2})$, suggesting an interesting interplay with our work.
However, very little is known when the learner has to cope \emph{simultaneously} with delayed and payoff-based feedback.

In the \ac{MAB} setting, the work of \citet{JGS13} provides an answer for mixed-strategy learning over finite-action spaces, but the online convex optimization case is completely different.
In particular, a major difficulty that arises is that the batch update approach of \citet{QK15} cannot be easily applied with stochastic estimates of the received gradient information (or when attempting to infer such information from realized payoffs).
This issue was highlighted in the work of \citet{ZMBG+17-NIPS} who employed a batching strategy similar to that of \citet{QK15} in a game-theoretic context with \emph{perfect} gradient information.
Because of this, online learning in the presence of delayed reward/feedback structures requires new tools and techniques.

On the game theory side, \citet{KKDB15} and \citet{BKTB16} studied the \acl{NE} convergence properties of no-regret learning in specific classes of continuous games (zero-sum and potential games).
The work of \citet{MZ19} and its follow-ups \citep{ZMMB+17,ZMMB+17-CDC,ZMAB+18-NIPS,lin2020finite} provided an extension to the class of monotone games with varying degrees of generality;
however, all these works rely on the availability of gradients in the learning process.
In sharp contrast to this, \citet{BBF18} recently considered payoff-based learning in games with one-dimensional action sets, and they established convergence to \acl{NE} under a synchronous, two-point, ``sample-then-play'' bandit strategy.
More recently, \citet{BLM18} showed that no-regret learning with payoff-based feedback converges to \acl{NE} in strongly monotone games, but it is assumed that actions are synchronized across players and rewards are assumed to arrive instantaneously.
A model of learning with delays was provided by \citet{ZMBG+17-NIPS} but their analysis and learning strategy only applies to perfect gradient information:
the case of noisy \textendash\ or, worse, \emph{payoff-based} \textendash\ delayed feedback was stated in that paper as a challenging open issue.
Our paper settles this open question in the affirmative.

\section{The model}
\label{sec:model}

\subsection{The general framework}

The general online optimization framework that we consider can be represented as the following sequence of events (presented for the moment from the viewpoint of a single, focal agent):
\begin{itemize}
\item
At each stage $\run=\running$, of the process, the agent picks an \emph{action} $\state_{\run}$ from a compact convex subset $\points$ of a $\vdim$-dimensional real space $\vecspace$.
\item
The choice of action generates a \emph{reward} $\est\pay_{\run} = \pay_{\run}(\state_{\run})$ based on a concave function $\pay_{\run}\from\points\to\R$ (assumed unknown to the player at stage $\run$).
\item
Simultaneously, $\state_{\run}$ triggers a \emph{delay} $\delay_{\run} \geq 0$ which determines the round $\run + \delay_{\run}$ at which the generated reward $\est\pay_{\run}$ will be received.
\item
The agent receives the rewards from all previous rounds $\info_{\run} = \setdef{\runalt}{\runalt + \delay_{\runalt} = \run}$, and the process repeats.
\end{itemize}

The above model has been stated in an abstract way that focuses on a single agent so as to provide the basis for the analysis to come.
The setting where there are no assumptions on the process generating the agent's payoff functions will be referred to as the \emph{unilateral} setting;
by contrast, in the multi-agent, \emph{game-theoretic} setting, the payoff functions of the focal agent will be determined by the stream of actions of the other players (see below for the details).
In the latter case, all variables other than the running counter $\run$ will be indexed by $\play$ to indicate their dependence on the $\play$-th player;
for example,
the action space of the $\play$-th player will be written $\points^{\play}$,
the corresponding action chosen by at stage $\run$ will be denoted $\state_{\run}^{\play}$,
etc.
For concreteness, we provide a diagrammatic illustration in \cref{fig:delays} above.


\begin{figure}[tbp]
\begin{center}
\small

\begin{tikzpicture}
[scale=1.05,
nodestyle/.style={circle,draw=black,fill=gray!10, inner sep=2pt},
edgestyle/.style={-},
>=stealth]

\def\dx{1}
\def\dy{.6}

\coordinate (T0) at (0,0);
\coordinate (X0) at (0,-\dy);
\coordinate (D0) at (0,-2*\dy);
\coordinate (F0) at (0,-3*\dy);

\coordinate (T1) at (1,0);
\coordinate (X1) at (1,-\dy);
\coordinate (D1) at (1,-2*\dy);
\coordinate (F1) at (1,-3*\dy);

\coordinate (T2) at (2,0);
\coordinate (X2) at (2,-\dy);
\coordinate (D2) at (2,-2*\dy);
\coordinate (F2) at (2,-3*\dy);

\coordinate (T3) at (3,0);
\coordinate (X3) at (3,-\dy);
\coordinate (D3) at (3,-2*\dy);
\coordinate (F3) at (3,-3*\dy);

\coordinate (T4) at (4,0);
\coordinate (X4) at (4,-\dy);
\coordinate (D4) at (4,-2*\dy);
\coordinate (F4) at (4,-3*\dy);

\coordinate (T5) at (5,0);
\coordinate (X5) at (5,-\dy);
\coordinate (D5) at (5,-2*\dy);
\coordinate (F5) at (5,-3*\dy);

\coordinate (T6) at (6,0);
\coordinate (X6) at (6,-\dy);
\coordinate (D6) at (6,-2*\dy);
\coordinate (F6) at (6,-3*\dy);

\coordinate (T7) at (7,0);
\coordinate (X7) at (7,-\dy);
\coordinate (D7) at (7,-2*\dy);
\coordinate (F7) at (7,-3*\dy);


\coordinate (Tlast) at (7,0);
\coordinate (Xlast) at (7,-\dy);
\coordinate (Dlast) at (7,-2*\dy);
\coordinate (Flast) at (7,-3*\dy);

\node (T0) at (T0) {$\run$};
\node (X0) at (X0) [text = DarkBlue] {$\state_{\run}$};
\node (D0) at (D0) [text = DarkMagenta] {$\delay_{\run}$};
\node (F0) at (F0) [text = DarkGreen] {$\rewards_\run$};

\node (T1) at (T1) [nodestyle] {$1$};
\node (X1) at (X1) [text = DarkBlue] {$\state_{1}$};
\node (D1) at (D1) [text = DarkMagenta] {$0$};
\node (F1) at (F1) [text = DarkGreen] {$\{\est\pay_{1}\}$};
\draw [edgestyle,->,loop above] (T1) to (T1);

\node (T2) at (T2) [nodestyle] {$2$};
\node (X2) at (X2) [text = DarkBlue] {$\state_{2}$};
\node (D2) at (D2) [text = DarkMagenta] {$3$};
\node (F2) at (F2) [text = DarkGreen] {$\{\}$};

\node (T3) at (T3) [nodestyle] {$3$};
\node (X3) at (X3) [text = DarkBlue] {$\state_{3}$};
\node (D3) at (D3) [text = DarkMagenta] {$1$};
\node (F3) at (F3) [text = DarkGreen] {$\{\}$};

\node (T4) at (T4) [nodestyle] {$4$};
\node (X4) at (X4) [text = DarkBlue] {$\state_{4}$};
\node (D4) at (D4) [text = DarkMagenta] {$2$};
\node (F4) at (F4) [text = DarkGreen] {$\{\est\pay_{3}\}$};
\draw [edgestyle,->,bend left] (T3.north east) to (T4.north west);

\node (T5) at (T5) [nodestyle] {$5$};
\node (X5) at (X5) [text = DarkBlue] {$\state_{5}$};
\node (D5) at (D5) [text = DarkMagenta] {$0$};
\node (F5) at (F5) [text = DarkGreen]  {$\{\est\pay_{2},\est\pay_{5}\}$};
\draw [edgestyle,->,bend left] (T2.north east) to (T5.north west);
\draw [edgestyle,->,loop above] (T5) to (T5);

\node (T6) at (T6) [nodestyle] {$6$};
\node (X6) at (X6) [text = DarkBlue] {$\state_{6}$};
\node (D6) at (D6) [text = DarkMagenta] {$1$};
\node (F6) at (F6) [text = DarkGreen] {$\{\est\pay_{4}\}$};
\draw [edgestyle,->,bend left] (T4.north east) to (T6.north west);



\node (Tlast) at (Tlast) {$\dots$};
\node (Xlast) at (Xlast) {$\dots$};
\node (Dlast) at (Dlast) {$\dots$};
\node (Flast) at (Flast) {$\dots$};
\draw [edgestyle,->,bend left] (T6.north east) to (Tlast.north west);

%
%

\end{tikzpicture}
\end{center}
\caption{Schematic illustration of the delayed feedback framework considered in the paper.
Arrows illustrate the round to which the payoff is deferred.}
\label{fig:delays}
\end{figure}
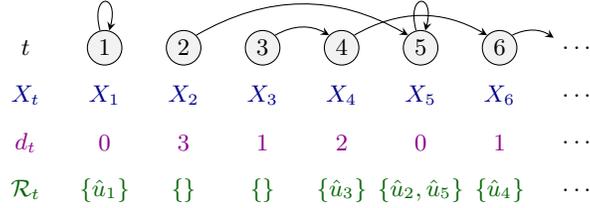


In both cases, our blanket assumptions for the stream of payoffs and the delays encountered by the players will be as follows:

\begin{assumption}
\label{asm:pay}
For each $\run=\running$, $\pay_{\run}$ is
concave in $\point$,
$\vbound$-Lipschitz continuous,
and
$\smooth$-Lipschitz smooth.
Specifically, the gradient $\payv_{\run}(\point) \equiv \nabla_{\point} \pay_{\run}(\point)$ of $\pay_{\run}$ is bounded by $\vbound$ and satisfies $\norm{\payv_{\run}(\pointalt) - \payv_{\run}(\point)} \leq \smooth \norm{\pointalt - \point}$ for all $\point,\pointalt\in\points$.
\end{assumption}

\begin{assumption}
\label{asm:delay}
The delays $\delay_{\run}$ grow asymptotically as $\delay_{\run} = o(\run^{\alpha})$ for some $\alpha<1$.
\end{assumption}

Regarding the delay assumption above, large-scale analytic studies have shown that long delays \emph{are} observed in practice:
in a study by \citet{Chapelle} with data from the real-time bidding company Criteo, it was found that more than $10\%$ of the conversions were more than two weeks old.
Moreover, the conclusion of the same study was that the distribution of delays in online advertising can be fitted reasonably well by long-tailed distributions,
especially when conditioning on context and feature variables available to the advertiser, thus justifying the assumption of a possibly unbounded delay between choosing an action and receiving a reward.
We also note here that we are making no further assumptions on the way the sequence of delays is generated:
conceivably, delays could even be determined adversarially, as in \citet{QK15}.

\subsection{Multi-agent considerations}

For the multi-agent case, suppose there is a finite set of players $\players = \{1,\dotsc,\nPlayers\}$, each with their own action space $\points^{\play} \subseteq \R^{\vdim^{\play}}$ (always assumed convex and compact).
In this case, it will be convenient to encode the players' joint action profile $\point = (\point^{\play})_{\play\in\players} \in \points \equiv \prod_{\play\in\players} \points^{\play}$ by means of the shorthand $(\point^{\play};\point^{-\play}) \equiv (\point^{1},\dotsc,\point^{\play},\dotsc,\point^{\nPlayers})$ which highlights the action $\point^{\play}\in\points^{\play}$ of the $\play$-th player against the action profile $\point^{-\play} \in \points^{-\play} \equiv \prod_{\playalt\neq\play} \points^{\playalt}$ of $\play$'s opponents.
The payoff to each player $\play\in\players$ for a given action profile $\point \in \points$ will then be determined by an associated \emph{payoff \textpar{\emph{or} utility} function} $\pay^{\play}\from\points\to\R$, assumed here and throughout to be concave in the action variable $\point^{\play}$ of the player in question.
We will refer to the tuple $\game \equiv \game(\players,\points,\pay)$ as an \emph{$\nPlayers$-player continuous game} \citep{Deb52,Ros65,FT91}.

In this context, if $\state_{\run} = (\state_{\run}^{1},\dotsc,\state_{\run}^{\nPlayers}) \in \points$ is a sequence of joint actions, the payoff function encountered by the $\play$-th player at stage $\run$ will be given by
\begin{equation}
\label{eq:pay-play}
\pay_{\run}^{\play}(\point^{\play})
	\equiv \pay^{\play}(\point^{\play};\state_{\run}^{-\play})
	\quad
	\text{for all $\point^{\play} \in \points ^{\play}$},
\end{equation}
leading to the gradient expression
\begin{align}
\label{eq:payv-t}
\payv_{\run}^{\play}(\point^{\play})
	&\equiv \nabla_{\point^{\play}} \pay_{\run}^{\play}(\point^{\play};\state_{\run}^{-\play})
	= \payv^{\play}(\point^{\play};\state_{\run}^{-\play})
\intertext{where}
\label{eq:payv}
\payv^{\play}(\point)
	&= \nabla_{\point^{\play}} \pay^{\play}(\point^{\play};\point^{-\play}).
\end{align}
denotes the individual payoff gradient of the $\play$-th player at the action profile $\point\in\points$.
In the rest of our paper, we will assume that $\pay^{\play}$ is Lipschitz continuous and Lipschitz smooth, so \cref{asm:pay} is satisfied by default in this case.

\subsection{Regret and equilibrium}

With these preliminaries at hand, our principal performance indicators will be the minimization of \emph{regret} and the notion of a \acli{NE}.
Starting with the former, the regret of an agent in the unilateral setting is defined over a horizon of $\nRuns$ stages as
\begin{equation}
\label{eq:regret}
\reg(\nRuns)
	= \max_{\point\in\points} \sum_{\run=1}^{\nRuns} \bracks{\pay_{\run}(\point) - \pay_{\run}(\state_{\run})}.
\end{equation}
and, in the presence of randomness, we similarly introduce the agent's mean (or pseudo-) regret as 
\begin{equation}
\label{eq:reg-mean}
\overline\reg(\nRuns)
	= \max_{\point\in\points}
		\exof{\sum_{\run=1}^{\nRuns}
			\bracks{\pay_{\run}(\point) - \pay_{\run}(\state_{\run})}}.
\end{equation}
Accordingly, we will say that a sequence of actions $\state_{\run}\in\points$, $\run=\running$, leads to \emph{no regret} if $\reg(\nRuns) = o(\nRuns)$.

On the other hand, the notion of a \acdef{NE} is a purely game-theoretic concept which characterizes those action profiles that are resilient to unilateral deviations.
In more detail, we say that $\eq \in \points$ is a \acl{NE} of $\game$ when
\begin{equation}
\label{eq:NE}
\tag{NE}
\pay^{\play}(\eq)
	\geq \pay^{\play}(\point^{\play};\eq^{-\play})
\end{equation}
for all $\point^{\play}\in\points^{\play}$ and all $\play\in\players$.
In full generality, the relation between \aclp{NE} and regret minimization is feeble at best:
if all players play a \acl{NE} for all $\run=\running$, they will trivially have no regret;
the converse however fails by a longshot, see \eg \citet{VZ13} and references therein.%
\footnote{Specifically, \citet{VZ13} show that ther are games whose set of coarse correlated equilibria contain strategies that assign positive probability \emph{only} to strictly dominated strategies.}

In the game-theoretic literature, existence and uniqueness of equilibrium points has been mainly studied under a condition known as \acdef{DSC} \citep{Ros65}, which we define here as:
\begin{equation}
\label{eq:DSC}
\tag{DSC}
\sum_{\play\in\players}
	\lambda^{\play} \braket{\payv^{\play}(\pointalt) - \payv^{\point}(\point)}{\pointalt^{\play} - \point^{\play}}
	< 0
\end{equation}
for some $\lambda^{\play}>0$ and all $\point,\pointalt\in\points$ with $\pointalt\neq\point$.

In optimization, this condition is known as \emph{monotonicity} \citep{BC17}, so we will interchange the terms ``diagonally strictly concave'' and ``\emph{monotone}''for games that satisfy \eqref{eq:DSC}.
Under \eqref{eq:DSC}, \citet{Ros65} showed the existence of a unique \acl{NE};
this is of particular importance to online advertising because of the following auction mechanism that can be seen as a monotone game:

\medskip
\begin{example}[Kelly auctions]
\label{ex:Kelly}
Consider a provider with a splittable commodity (such as advertising time or website traffic to which a given banner will be displayed).
Any fraction of this commodity can be auctioned off to a set of $\nPlayers$ bidders (players) who can place monetary bids $\point^{\play}\geq0$ up to each player's total budget $b^{\play}$ to acquire it.
Once all players have placed their respective bids, the commodity is split among the bidders proportionally to each player's bid;
specifically, the $\play$-th player gets a fraction
\(
\rho^{\play}
	= \point^{\play} \big/ \parens{c + \sum_{\playalt\in\players} \point_{\playalt}}
\)
of the auctioned commodity (where $c\geq0$ is an ``entry barrier'' for bidding on the resource).
A simple model for the utility of player $\play$ is then given by the \emph{Kelly auction mechanism} \citep{KMT98}:
\begin{equation}
\label{eq:pay-auction}
\pay^{\play}(\point^{\play};\point^{-\play})
	= g^{\play} \rho^{\play} - \point^{\play},
\end{equation}
where
$g^{\play}$ represents the marginal gain of player $\play$ from a unit of the commodity.
Using standard arguments, it is easy to show that the resulting game satisfies \eqref{eq:DSC}, see \eg \citet{Goo80}.
\end{example}

Other example of games satisfying \eqref{eq:DSC} are (strictly) convex-concave zero-sum games \citep{JNT11}, routing games \citep{NRTV07}, Cournot oligopolies \citep{MZ19}, power control \citep{SFPP10,MBNS17}, etc.
For an extensive discussion of monotonicity in game theory, see \citet{FK07,PSPF10,San15,LRS19} and references therein.
In the rest of our paper, we will assume that all games under consideration satisfy \eqref{eq:DSC}.

\section{The \ac{GOLD} algorithm}
\label{sec:method}

We are now in a position to state the proposed \acdef{GOLD} method.
As the name suggests, the method concurrently addresses the two aspects of the online learning framework presented in the previous section, namely the delays encountered and the lack of gradient information.
We describe each component in detail below, and we provide a pseudocode implementation of the method as \cref{alg:GOLD} above;
for convenience and notational clarity, we take the viewpoint of a focal agent throughout, and we do not carry the player index $\play$.


\begin{algorithm*}[tbp]
\tt
\small
\caption{\acf{GOLD}	\hfill{\footnotesize[focal player view]}}
\label{alg:GOLD}

\setstretch{1.125}
\begin{algorithmic}[1]
\Require
	step-size $\step_{\run}>0$,
	sampling radius $\mix_{\run}>0$,
	safety set $\ball_{\radius}(\base)\subseteq\points$
\State
	choose $\state_{\start}\in\points$;
	set $\pool_{0} \leftarrow \varnothing$,
	$\est\pay_{\infty} = 0$,
	$\unitvar_{\infty} = 0$
	\Comment{initialization}%
\For{$\run=\running$}
	\State	
		draw $\unitvar_{\run}$ uniformly from $\sphere^{\vdim}$
		\Comment{perturbation direction}%
	\State
		set $\pertvar_{\run}\leftarrow\unitvar_{\run} - (\state_{\run} - \base)/\radius$
		\Comment{feasibility adjustment}%
	\State
		play $\est\state_{\run} \leftarrow \state_{\run} + \mix_{\run} \pertvar_{\run}$
		\Comment{player chooses action}%
	\State
		generate payoff $\est\pay_{\run} = \pay(\est\state_{\run})$
		\Comment{associated payoff}%
	\State
		trigger delay $\delay_{\run}$
		\Comment{delay for payoff}%
	\State
		collect rewards $\info_{\run} = \setdef{\runalt}{\runalt + \delay_{\runalt} = \run}$
		\Comment{receive past payoffs}%
	\State
		update pool $\pool_{\run} \leftarrow \pool_{\run-1} \cup \info_{\run}$
		\Comment{enqueue received info}%
	\State
		take $\head_{\run} = \min \pool_{\run}$;
		set $\pool_{\run} \leftarrow \pool_{\run}\exclude{\head_{\run}}$
		\Comment{dequeue oldest info}%
	\State
		set $\est\payv_{\run}\leftarrow (\vdim/\mix_{\head_{\run}}) \est\pay_{\head_{\run}} \, \unitvar_{\head_{\run}}$
		\Comment{estimate gradient}%
	\State
		update $\state_{\run+1} \leftarrow \Eucl(\state_{\run} + \step_{\run} \est\payv_{\run})$
		\Comment{update pivot}%
\EndFor
\end{algorithmic}
\end{algorithm*}


\subsection{Delays}

To describe the way that the proposed method tackles delays, it is convenient to decouple the two issues mentioned above and instead assume that, at time $\run$, along with the generated rewards $\est\pay_{\runalt}$ for $\runalt \in \info_{\run} = \setdef{\runalt}{\runalt + \delay_{\runalt} = \run}$, the agent also receives \emph{perfect gradient information} for the corresponding rounds, \ie gets to observe $\payv_{\runalt}(\state_{\runalt})$ for $\runalt \in \info_{\run}$.
We stress here that this assumption is \emph{only} made to illustrate the way that the algorithm is handling delays, and will be dropped in the sequel.

With this in mind, the first thing to note is that the set of information received at a given round might be empty, \ie we could have $\info_{\run} = \varnothing$ for some $\run$.
To address this sporadic shortage of information, we introduce a \emph{pooling strategy}, not unlike the one considered by \citet{JGS13} in the context of \acl{MAB} problems.
Specifically, we assume that, as information is received over time, the agent adds it to an \emph{information pool} $\pool_{\run}$, and then uses the oldest information available in the pool (where ``oldest'' refers to the time at which the information was generated).

Specifically, starting at $\run=0$ with an empty pool $\pool_{0} = \varnothing$ (since there is no information at the beginning of the game), the agent's information pool is updated following the recursive rule
\begin{equation}
\label{eq:pool}
\pool_{\run}
	= \pool_{\run-1} \cup \info_{\run} \setminus \{\head_{\run}\}
\end{equation}
where
\begin{equation}
\label{eq:head}
\head_{\run}
	= \min(\pool_{\run-1} \cup \info_{\run})
\end{equation}
denotes the oldest round from which the agent has unused information at round $\run$.
Heuristically, this scheme can be seen as a priority queue in which data $\payv_{\runalt}(\state_{\runalt})$, $\runalt\in\info_{\run}$, arrives at time $\run$ and is assigned priority $\runalt$ (\ie the round from which the data originated);
subsequently, gradient data is dequeued one at a time, in ascending priority order (\ie oldest information is utilized first).

In view of the above, if we let
\(
\est\payv_{\run}
	= \payv_{\head_{\run}}(\state_{\head_{\run}})
\)
denote the gradient information dequeued at round $\run$, we will use the basic gradient update
\begin{equation}
\label{eq:GD-base}
\state_{\run+1}
	= \Eucl(\state_{\run} + \step_{\run} \est\payv_{\run}),
\end{equation}
where
$\step_{\run}>0$ is a variable step-size sequence (discussed extensively in the sequel),
and
$\Eucl(\dpoint) = \argmin_{\point\in\points} \norm{\point - \dpoint}$ denotes the Euclidean projection to the agent's action space $\points$.
Of course, an important issue that arises in the update step \eqref{eq:pool} is that, despite the parsimonious use of gradient information, it may well happen that the agent's information pool $\pool_{\run}$ is empty at time $\run$ (\eg if at time $\run=\start$, we have $\delay_{\start} > 0$).
In this case, following the standard convention $\inf\varnothing=\infty$, we set $\head_{\run} = \infty$ (since it is impossible to ever have information about the stage $\run=\infty$), and, by convention, we also set $\payv_{\infty} = 0$.
Under this convention, \eqref{eq:GD-base} can be written in more explicit form as
\begin{align}
\label{eq:GD-perfect}
\state_{\run+1}
	&= \Eucl(\state_{\run} + \step_{\run} \one_{\pool_{\run}\neq\varnothing}\est\payv_{\run})
	\notag\\
	&= \begin{cases}
		\state_{\run}
			&\quad
			\text{if $\pool_{\run} = \varnothing$},
			\\
		\Eucl(\state_{\run} + \step_{\run} \est\payv_{\run})
			&\quad
			\text{otherwise}.
		\end{cases}
\end{align}
In this way, the gradient update \eqref{eq:GD-base} can be seen as a delayed variant of Zinkevich's online gradient descent policy;
however, in contrast to ``batching-type'' policies \citep{QK15,ZMBG+17-NIPS}, there is no gradient aggregation:
received gradients are introduced in the algorithm one at a time, oldest information first.

\subsection{Payoff-based gradient estimation}

We now proceed to describe the process with which the agent infers gradient information from the received rewards.
To that end, following \citet{Spa97} and \citet{FKM05}, we will use a one-point, \acf{SPSA} approach that was also recently employed by \citet{BLM18} for game-theoretic learning with bandit feedback (but no delays or asynchronicities).
In our delayed reward setting (and always from the viewpoint of a single, focal agent), this process can be described as follows:
\begin{enumerate}
\item
Pick a pivot state $\state_{\run}$ to estimate its payoff gradient.
\item
Pick a sampling radius $\mix_{\run} > 0$ (detailed below) and draw a random sampling direction $\unitvar_{\run}$ from the unit sphere $\sphere^{\vdim}$.
\item
Introduce an adjustment $\pertvar_{\run}$ to $\unitvar_{\run}$ to ensure feasibility of the sampled action
\begin{equation}
\label{eq:action}
\hat\state_{\run}
	= \state_{\run} + \mix_{\run}\pertvar_{\run}
\end{equation}
\item
Generate the reward $\est\pay_{\run} = \pay_{\run}(\est\state_{\run})$ and estimate the gradient of $\pay_{\run}$ at $\state_{\run}$ as
\begin{equation}
\label{eq:SPSA}
\est\nabla_{\run}
	= \frac{\vdim}{\mix_{\run}} \est\pay_{\run} \unitvar_{\run}
\end{equation}
\end{enumerate}
More precisely, the feasibility adjustment mentioned above is a skewing operation of the form
\begin{equation}
\label{eq:pertvar}
\pertvar_{\run}
	= \unitvar_{\run} - \radius^{-1}(\state_{\run} - \base)
\end{equation}
where $\base\in\points$ and $\radius>0$ are such that the radius-$\radius$ ball $\ball_{\radius}(\base)$ has $\ball_{\radius}(\base) \subseteq \points$, ensuring in this way that $\est\state_{\run} \in \points$ whenever $\state_{\run}\in\points$;
for more details, see \citet{BCB12}.

\subsection{Learning with delayed, payoff-based feedback}

Of course, the main problem in the \ac{SPSA} estimator \eqref{eq:SPSA} lies in the fact that, in a delayed reward structure, the payoff generated at time $\run$ would only be observed at stage $\run+\delay_{\run}$.
With this in mind, we make the following bare-bones assumptions:
\begin{itemize}
\item
Expectations are taken relative to the inherent randomness in the sampling direction $\unitvar_{\run}$.
\item
The agent retains in memory the chosen sampling direction $\unitvar_{\runalt}$ for all $\runalt \leq \run$ that have not yet been utilized, \ie for all $\runalt \in \unused_{\run} \equiv \{1,\dotsc,\run\} \setminus \setdef{\head_{\ell}}{\ell=1,\dotsc,\run}$.%
\footnote{In the appendix, we show that $\abs{\unused_{\run}} \leq \max_{1\leq\runalt\leq\run} \delay_{\runalt}$, so this requirement is fairly mild (linear) relative to the delays, especially when the delay distribution is exponential \textendash\ \eg as in the online advertising study of \citet{Chapelle}.}
\end{itemize}
In this way, to combine the two frameworks described above (delays \emph{and} bandit feedback), we will employ the gradient estimator
\begin{equation}
\label{eq:oracle}
\est\payv_{\run}
	= \one_{\pool_{\run}\neq\varnothing} \est\nabla_{\head_{\run}}
	= \frac{\vdim}{\mix_{\head_{\run}}}
		\est\pay_{\head_{\run}}
		\,\unitvar_{\head_{\run}}
\end{equation}
with the convention $\est\pay_{\infty} = 0$, $\unitvar_{\infty} = 0$ if $\head_{\run} = \infty$ \textendash\ \ie if the player's information pool $\pool_{\run}$ is empty at stage $\run$.
Thus, putting everything together, we obtain the \acdef{GOLD} policy:
\begin{equation}
\label{eq:GOLD}
\tag{GOLD}
\begin{aligned}
\est\state_{\run}
	&= \state_{\run} + \mix_{\run} \pertvar_{\run}
	\\
\state_{\run+1}
	&= \Eucl(\state_{\run} + \step_{\run} \est\payv_{\run})	
\end{aligned}
\end{equation}
with $\pertvar_{\run}$ and $\est\payv_{\run}$ given by \cref{eq:pertvar,eq:oracle} respectively (for a pseudocode implementation of the policy, see \cref{alg:GOLD}).
We will examine the learning properties of this policy in the next section.


\section{Analysis and guarantees}
\label{sec:results}

\subsection{Statement and discussion of main results}

We are now in a position to state and prove our main results for the \ac{GOLD} algorithm under \cref{asm:delay,asm:pay}.
We begin with the algorithm's regret guarantees in the unilateral setting:

\begin{restatable}{theorem}{regret}
\label{thm:regret}
Suppose that an agent is running \eqref{eq:GOLD} with step-size and sampling radius sequences of the form
$\step_{\run} = \step/\run^{c}$ and $\mix_{\run} = \mix/\run^{b}$ for some $\step,\mix > 0$ and
$b = \min\{1/4,1/3-\alpha/3\}$,
$c = \max\{3/4,2/3+\alpha/3\}$.
Then, the agent enjoys the mean regret bound
\begin{equation}
\label{eq:reg-bound}
\overline\reg(\nRuns)
	= \tilde\bigoh\parens[\big]{\nRuns^{3/4} + \nRuns^{2/3 + \alpha/3}}.
\end{equation}
\end{restatable}

\begin{remark*}
In the above, $\tilde\bigoh(\cdot)$ stands for ``$\bigoh(\cdot)$ up to logarithmic factors''.
The actual multiplicative constants that are hidden in the Landau ``big oh'' notation have a complicated dependence on the diameter of $\points$, the dimension of the ambient space, the range of the players' utility functions;
we provide more details on this in the paper's appendix.
\end{remark*}

For the game-theoretic setting, we will focus on games satisfying Rosen's \acl{DSC} condition (\eg as the Kelly auction example described in \cref{sec:model}).
In this general context, we have:

\begin{restatable}{theorem}{Nash}
\label{thm:Nash}
Let $\game$ be a continuous game satisfying \eqref{eq:DSC}, and suppose that each agent follows \eqref{eq:GOLD} with step-size and sampling radius sequences
$\step_{\run} = \step/\run^{c}$ and $\mix_{\run} = \mix/\run^{b}$ for some $\step,\mix > 0$ and
$b,c$ satisfying the conditions:
\begin{subequations}
\label{eq:params}
\begin{align}
2c - b
	&> 1 + \alpha,
	\\
b + c
	&> 1,
	\\
2c - 2b
	&> 1.
\end{align}
\end{subequations}
Then, with probability $1$, the sequence of play $\est\state_{\run}$ induced by \eqref{eq:GOLD} converges to the game's \textpar{necessarily} unique \acl{NE}.
\end{restatable}

The above results are our main guarantees for \eqref{eq:GOLD} so, before discussing their proof, some remarks are in order.
The first concerns the tuning of the algorithm's hyperparameters, \ie the exponents $b$ and $c$.
Even though the conditions stated in \cref{thm:regret} may appear overly precise, we should note that agents have considerably more leeway at their disposal.
Specifically, as part of the proof, we show that any choice of the exponents $b$ and $c$ satisfying \eqref{eq:params} also leads to no regret \textendash\ albeit possibly at a worse rate.
This is particularly important for the interplay between no regret and convergence to \acl{NE} because it shows that the two guarantees are fairly well aligned as long as \eqref{eq:GOLD} is the class of no-regret policies under consideration.


\begin{figure}[t]
\begin{center}
\small

\begin{tikzpicture}
[scale=0.5,
nodestyle/.style={circle,draw=black,fill=gray!10, inner sep=2pt},
edgestyle/.style={-},
>=stealth]

\small

\def\dx{1}
\def\dy{.6}

\draw[thick,->] (0,-.5) -- (0,9.5) node[anchor=north east] {b};
\draw[thick,->](-.5,0) -- (9.5,0) node[anchor=north west] {c};

\fill[DarkGreen!40!white] (8,0) --(8,4)--(6,2);

\draw[thick,dashed,RoyalBlue] (9,-1) -- (-1,9);
\draw[thick,RoyalBlue] (8,-1) -- (8,9) ;
\draw[thick,dashed,RoyalBlue] (3,.-1) -- (9,5);

\draw[thick,Crimson] (7,-2) -- (9, 2);
\coordinate (A0) at (9.75, 2.5);
\node (A0) at (A0) [text = Crimson] {$\alpha = 1$};
\draw[thick,Crimson] (5,-2) -- (9, 6);
\coordinate (A1) at (9.75, 6.5);
\node (A1) at (A1) [text = Crimson] {$\alpha =  \frac{1}{2}$};
\draw[thick,Crimson] (4,-2) -- (9, 8);
\coordinate (A2) at (9.75, 8.5);
\node (A2) at (A2) [text = Crimson] {$\alpha =  \frac{1}{4}$};
\draw[thick,Crimson] (3,-2) -- (9,10);
\coordinate (A3) at (9.75, 10.5);
\node (A3) at (A3) [text = Crimson] {$\alpha = 0$};

\draw[step=1cm,dash dot,gray,very thin] (-.2,-.2) grid (8.7,8.7);

\coordinate (A4) at (6.5, 7);
\node (A4) at (A4) [text = Crimson,rotate=63.4349488] {$b \leq 2c - 1 - \alpha$};

\draw (6,0.2) -- (6,-0.2) node[anchor=north] {$\frac{3}{4}$};
\draw ( 8,0.2) -- (8,-0.2) node[anchor=north west] {$1$};
\draw (0.2,8) -- (-0.2,8) node[anchor=east] {$1$};

\end{tikzpicture}
\end{center}
\caption{The allowable region (green shaded areas) of possible values of the sampling radius and step-size exponents $b$ and $c$ for various values of the groth exponent $\alpha$ of the encountered delays.
The dashed blue lines corresponding to the last two terms in \eqref{eq:params} indicate hard boundaries leading to logarithmic terms in the regret instead of constants.}
\label{fig:abc}
\end{figure}
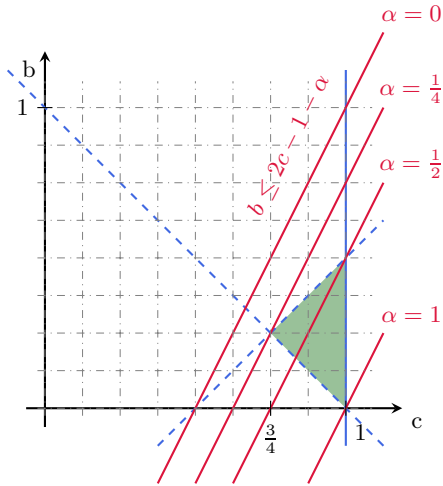


We should also note here that
the $\nRuns^{3/4}$ term is the standad regret bound that one obtains in the bandit online convex optimization framework.
On the other hand, the term $\nRuns^{2/3+\alpha/3}$ describes the advent of the delays which, combined with the bias of the \ac{SPSA} gradient estimator, contribute a significant amount of regret over time (recall in particular that $\delay_{\run}$ is a priori unbounded).
This is of particular importance to applications to online advertising where delays can often become arbitrarily large.

For concreteness, we also plot in \cref{fig:abc} the region of allowed step-size and sampling radius exponents.
This plot reveals the interesting property that, if the feedback delays do not grow too large over time \textendash\ specifically, if $\delay_{\run} = o(\run^{1/4})$ \textendash\ then they have no impact on the allowable choices of $b$ and $c$.
This is also reflected in the regret bound \eqref{eq:reg-bound} where, for $\alpha = 1/4$, the regret-specific term becomes $\nRuns^{3/4}$ as well;
in particular, in the \emph{constant regret} case $\delay_{\run} = \bigoh(1)$, the delays are invisible in \eqref{eq:reg-bound}.
These considerations illustrate the impact of each source of feedback scarcity (bandit vs. delays) on the performance of \eqref{eq:GOLD} and provides a clear insight on the different mechanisms affecting the algorithm's regret and convergence guarantees.

\subsection{Analysis and sketch of proof}

The rest of this section is devoted to a high-level sketch of the proof of \cref{thm:regret,thm:Nash}.
We begin by decomposing $\est{\payv}_{\run}$ as a noisy estimate of $ \payv(\state_{\head_{\run}})$ into the following elements:
\begin{equation}
\label{eq:signal}
\est{\payv}_{\run}
	= \payv(\state_{\head_{\run}}) +\noise_{\head_{\run}+1} + \bias_{\head_{\run}}.
\end{equation}
The various terms in \eqref{eq:signal} above are defined as follows:
\begin{enumerate}
\item
First, we set
\begin{equation}
\label{eq:noise}
\noise_{\head_{\run}+1}
	= \est{\payv}_{\run} -  \exof{\est{\payv}_{\run}\given \filter_{\run}}
\end{equation}
where the filtration $\filter_{\run}$ contains all the random variables that have been realized at the beginning of the $\run$-th iteration of the algorithm;
more precisely, we let
\begin{equation}
\label{eq:filter}
\filter_{\run}
	= \sigma(
		\varnothing,
		\state_{\start},
		\dotsc,
		\est\pay_{\head_{\run-1}},\unitvar_{\run-1},\state_{\run})
\end{equation}
with the convention $\est\pay_{\infty} = 0$, $\unitvar_{\infty} = 0$ if $\head_{\run} = \infty$.
We note for posterity that $\noise_{\head_{\run}}$ is a martingale difference sequence relative to $\filter_{\run}$, \ie $\exof{\noise_{\head_{\run}+1} \given \filter_{\run}} = 0$.

\item
Second, we let
\begin{equation}
\bias_{\head_{\run}}
	= \exof{\est\payv_{\run} \given \filter_{\run}} - \payv(\state_{\head_{\run}})
\end{equation}
denote the systematic error of the estimator $\est\payv_{\run}$ relative to the gradient of the dequeued state $\state_{\head_{\run}}$ (\ie the error remaining after any zero-sum component has been averaged out).
In contrast to $\noise$, this term is not zero-mean;
instead, as we discuss in the appendix, the \ac{SPSA} gradient estimation process that we employ induces a bias of order $\norm{\bias_{\head_{\run}}} = \bigoh(\mix_{\head_{\run}})$.
This bias term grows smaller with $\run$ but its variance increases, leading to a bias-variance trade-off in our setting.
\end{enumerate}

With all this in hand, the workhorse of our calculations is the distance of the sequence $\state_{\run}$ to a given ``benchmark'' action $\base\in\points$ (the best fixed action in hindsight, or the game's equilibrium, depending on the context).
Specifically, letting
\begin{equation}
\label{eq:energy}
\breg_{\run}
	= \frac{1}{2} \norm{\state_{\run} - \base}^{2}
\end{equation}
we have the following template inequality:

\begin{restatable}{lemma}{template}
\label{lem:template} 
If \eqref{eq:GOLD} is run with assumptions as above, then, for all $\base\in\points$, we have
\begin{subequations}
\label{eq:template}
\begin{align}
\breg_{\run+1}
	\leq \breg_{\run}
	&\label{eq:template-drift}
		+ \step_{\run} \braket{\payv(\state_{\head_{\run}})}{\state_{\run} - \base}
	\\
	&\label{eq:template-noise}
		+ \step_{\run} \braket{\noise_{\head_{\run}+1}}{\state_{\run} - \base}
	\\
	&\label{eq:template-bias}
		+ \step_{\run} \braket{\bias_{\head_{\run}}}{\state_{\run} - \base}
	\\
	&\label{eq:template-error}
		+ \tfrac{1}{2} \step_{\run}^{2} \norm{\hat\payv_{\run}}^{2}
\end{align}
\end{subequations}
\end{restatable}

This lemma follows from the decomposition \eqref{eq:signal}, the nonexpansivity of the projection mapping, and the regularity assumption \eqref{asm:pay} which allows us to control the terms \eqref{eq:template-drift} and \eqref{eq:template-bias} above;
to streamline our discussion, we defer the details to the paper's supplement.
Moving forward, with this estimate at our disposal, the analysis branches for \cref{thm:regret,thm:Nash} as indicated below.

\para{Regret analysis}
To bound the agent's regret, we need to isolate the scalar product in \eqref{eq:template-drift} and telescope through $\run=\running,\nRuns$ after dividing by the step-size $\step_{\run}$.
Deferring the ensuing lengthy calculations to the appendix, we ultimately obtain a bound of the form
\begin{equation}
\label{eq:reg-bound-sum}
\overline\reg(\nRuns)
	= \bigoh\parens*{
		\frac{1}{\step_{\nRuns}}
			\sum_{\run=\start}^{\nRuns}
			\parens*{
				\step_{\run} \sum_{\runalt=\head_{\run}}^{\run-1} \frac{\step_{\runalt}}{\mix_{\runalt}}
				+ \step_{\run} \mix_{\head_{\run}}
				+ \frac{\step_\run^2}{\mix_{\head_{\run}}^2}
			} 
		} 
\end{equation}
As a result, to proceed, we need to provide a specific bound for each of the above summands.
The difficulty here is the mixing of different quantities at different time-stamps, \eg as in the product term $\step_{\run} \mix_{\head_{\run}}$.
Bounding these terms requires a delicate analysis of the delay terms in order to estimate the maximum distance between $\run$ and $\head_{\run}$.
We will return to this point below;
for now, with some hindsight, we only stress that the terms in \eqref{eq:reg-bound-sum} correspond on a one-to-one basis with the conditions \eqref{eq:params} for the parameters of \eqref{eq:GOLD}.

\para{Game-theoretic analysis}

The game-theoretic analysis is significantly more involved and relies on a two-pronged approach:
\begin{enumerate}
\item
We first employ a version of the Robbins\textendash Siegmund theorem to show that the random variable $\breg_{\run} = (1/2) \norm{\state_{\run} - \eq}^{2}$ converges pointwise as $\run\to\infty$ to a random variable $\breg_{\infty}$ that is bounded in expectation (here $\eq$ denotes the game's unique equilibrium).

\item
Subsequently, we use a series of probabilistic arguments (more precisely, a law of large numbers for martingale difference sequences and Doob's submartingale convergence theorem) to show that \eqref{eq:GOLD} admits a (possibly random) subsequence $\state_{\run_{\runalt}}$ converging to $\eq$.
\end{enumerate}
Once these two distinct elements have been obtained, we can readily deduce that $\state_{\run} \to \eq$ with probability $1$ as $\run\to\infty$.
Hence, given that $\norm{\state_{\run} - \est\state_{\run}} = \bigoh(\mix_{\run})$ and $\lim_{\run} \mix_{\run} = 0$, our claim would follow.

However, applying the probabilistic arguments outlined above requires in turn a series of summability conditions.
Referring to the paper's supplement for the details, these requirements boil down to showing that the sequences
\begin{equation}
\label{eq:3seqs}
A_{\run}
	= \step_{\run} \sum_{\runalt=\head_{\run}}^{\run-1} \frac{\step_{\runalt}}{\mix_{\runalt}},
	\quad
B_{\run}
	= \step_{\run} \mix_{\head_{\run}},
	\quad
\textrm{and}
	\quad
C_{\run}
	= \frac{\step_\run^2}{\mix_{\head_{\run}}^2},
\end{equation}
are all summable.
Importantly, each of these three sums has a clear and concise interpretation in our learning context:
\begin{enumerate}
\item
The first term ($A_{\run}$) is the cumulative error induced by using outdated information.
\item
The second term ($B_{\run}$) is the error propagated from the bias of the \ac{SPSA} estimator.
\item
Finally, the third term ($C_{\run}$) corresponds to the variance (or, rather, the mean square) of the \ac{SPSA} estimator.
\end{enumerate}
As a result, as long as these terms are all summable, their impact on the learning process should be relatively small (if not outright negligible).

Comparing the above term-by-term to \eqref{eq:reg-bound-sum} is where the game-theoretic analysis rejoins the regret analysis.
As we said above, this requires a careful treatment of the delay process, which we outline below.

\para{Delay analysis}

A key difficulty in bounding the sums in \eqref{eq:reg-bound-sum} is that the first term ($A_{\run}$ in \eqref{eq:3seqs} is a sum of $\run-\head_{\run}$ terms, so it can grow quite rapidly in principle.
However, our pooling strategy guarantees that $\run-\head_{\run}$ cannot grow faster than the delay (which is sublinear by assumption).
This observation (detailed in the supplement) guarantees the convergence of the sum.
A further hidden feature of \eqref{eq:template} is in the noise term $\noise_{\run}$:
in the case of batching or reweighted strategies (\eg as in \citealp{ZMBG+17-NIPS}), this term incorporates a sum of terms arriving from different stages of the process, making it very difficult (if not impossible) to control.
By contrast, the pooling strategy that defines the \ac{GOLD} policy allows us to treat this as an additional ``noise'' variable;
we achieve this by carefully choosing the step-size and sampling radius parameters based on the following lemma:

\begin{restatable}{lemma}{threesums}
\label{lem:3sums}
Suppose that \eqref{eq:GOLD} is run with step-size and sampling radius parameters of the form $\step_{\run} \propto \step/\run^{c}$ and $\mix_{\run} \propto \mix/\run^{b}$, with $b,c>0$.
Then:
\begin{enumerate}
\item
If $2c-b \geq 1+\alpha$, then $\sum_{\run=\start}^{\nRuns} A_{t} = \bigoh(\log\nRuns)$;
in addition, if the inequality is strict, $A_{\run}$ is summable.
\item
If $c+b\geq 1$, then $\sum_{\run=\start}^{\nRuns} B_{t} = \bigoh(\log\nRuns)$;
in addition, if the inequality is strict, $B_{\run}$ is summable.
\item
If $2c - 2b \geq 1$, then $\sum_{\run=\start}^{\nRuns} C_{t} = \bigoh(\log\nRuns)$;
in addition, if the inequality is strict, $C_{\run}$ is summable.
\end{enumerate}
\end{restatable}
Proving this lemma requires a series of intermediate results that we defer to the paper's supplement.

\section{Concluding remarks}
\label{sec:conclusion}

Our aim in this paper was to examine the properties of bandit online learning in games with continuous action spaces and a delayed reward structure (with a priori unbounded delays).
The proposed \ac{GOLD} policy is the first in the literature to simultaneously achieve no regret and convergence to \acl{NE} with delayed rewards \emph{and} bandit feedback.
From a regret perspective, it matches the standard $\bigoh(\nRuns^{3/4})$ bound of \citet{FKM05} if the delay process is tame (specifically, if $\delay_{\run}$ grows no faster than $o(\run^{1/4})$);
in addition, from game-theoretic standpoint, it converges to equilibrium with probability $1$ in all games satisfying Rosen's \ac{DSC} condition.

One important direction for future research concerns the case of anonymous \textendash\ \ie not time-stamped \textendash\ rewards.
This complicates the matters considerably because it is no longer possible to match a received reward to an action;
as a result, the \ac{GOLD} policy would have to be redesigned from the ground up in this context.
Another important avenue is that
the kernel-based estimation techniques of \citet{BE16,BE17} achieve a faster $\bigoh(\nRuns^{1/2})$ regret minimization rate with bandit feedback;
whether this is still achievable with a delayed reward structure, and whether this can also lead to fast(er) convergence to \acl{NE} is another direction for future research.

\appendix
\numberwithin{equation}{section}		
\numberwithin{lemma}{section}		
\numberwithin{proposition}{section}		
\numberwithin{theorem}{section}		

\section{Auxiliary results}
\label{app:aux}

In this appendix, we collect some basic results for the \ac{SPSA} gradient estimator and the gradient update step in \eqref{eq:GOLD}.
We begin by establishing the template inequality of \cref{lem:template} which, for convenience, we restate below:

\template*

\begin{proof}
We begin by recalling the decomposition of $\est\payv_{\run}$ as
\begin{equation}
\tag{\ref*{eq:signal}}
\est{\payv}_{\run}
	= \payv(\state_{\head_{\run}}) +\noise_{\head_{\run}+1} + \bias_{\head_{\run}}.
\end{equation}
where:
\begin{enumerate}
\item
the noise process $\noise_{\head_{\run}+1} = \est\payv_{\run} -  \exof{\est\payv_{\run}\given \filter_{\run}}$ is a zero-mean process when conditioned on the filtration $\filter_{\run} = \sigma(\varnothing,\state_{\start},\dotsc,\est\pay_{\head_{\run-1}},\unitvar_{\run-1},\state_{\run})$ that contains all random variables that have been realized at the beginning of the $\run$-th iteration of the algorithm.
\item
$\bias_{\head_{\run}} = \exof{\est\payv_{\run} \given \filter_{\run}} - \payv(\state_{\head_{\run}})$ denotes the systematic (non-zero-mean) error of the estimator $\est\payv_{\run}$ relative to the gradient of the dequeued state $\state_{\run}$.
\end{enumerate}
We note for posterity that $\noise_{\head_{\run}}$ is a martingale difference sequence relative to $\filter_{\run}$, \ie $\exof{\noise_{\head_{\run}+1} \given \filter_{\run}} = 0$.
We also note here that, when the pool is empty, there is no update so $\state_{\run+1} = \state_{\run}$.
Then, for any $\base\in\points$ and for all $\run=\running\nRuns$ for which an update occurs, we have
\begin{align}
\label{eq:start}
\norm{\state_{\run+1} - \base}^{2}
	&= \norm{\Eucl(\state_{\run} + \step_{\run}\est\payv_{\run}) - \base}^{2}
	\notag\\
	&= \norm{\Eucl(\state_{\run} + \step_{\run}\est\payv_{\run}) - \Eucl(\base)}^{2}
	\notag\\
	&\leq \norm{\state_{\run} + \step_{\run}\est\payv_{\run} - \base}^{2}
	\notag\\
	&\leq \norm{\state_{\run} - \base}^{2}
		+ 2\step_{\run} \braket{\est\payv_{\run}}{\state_{\run} - \base}
		+ \step_{\run}^{2} \norm{\est\payv_{\run}}^{2}
	\notag\\
	&= \norm{\state_{\run} - \base}^{2}
		+ 2\step_{\run} \braket{\payv(\state_{\head_{\run}}) +\noise_{\head_{\run}+1} + \bias_{\head_{\run}} }{\state_{\run} - \base}
		+ \step_{\run}^{2} \norm{\est\payv_{\run}}^{2}
		\notag\\
	&= \norm{\state_{\run} - \base}^{2}
		+ 2\step_{\run}\braket{\payv(\state_{\run})}{\state_{\run} - \base}
	\notag\\
	&\hphantom{= \norm{\state_{\run} - \base}^{2}\;}
		+ 2\step_{\run} \braket{\noise_{\head_{\run}+1}}{\state_{\run} - \base}
		+ 2\step_{\run} \braket{\bias_{\head_{\run}}}{\state_{\run} - \base}
		+ \step_{\run}^{2} \norm{\hat\payv_{\run}}^{2}.
\end{align}
Our claim then follows by recalling that $\breg_{\run} = (1/2) \norm{\state_{\run} - \base}^{2}$.
\end{proof}

Coupled with \cref{lem:template}, the decomposition \eqref{eq:signal} will allow us to control the distance to a chosen bencmark action $\base$ by properly bounding each of the summands of \eqref{eq:template}.
To that end, we provide below a series of estimates for each of these terms:

\begin{lemma}
\label{lem:snoise}
Let $\snoise_{\head_{\run}+1} = \braket{\noise_{\head_{\run}+1}}{\state_{\run} - \base}$.
Then:
\begin{equation}
\exof{\snoise_{\head_{\run}+1}}
	= 0
	\quad
	\text{for all $\run=\running$}
\end{equation}
\end{lemma}

\begin{proof}
By the law of total expectation, we have:
\begin{align}
\exof{\snoise_{\head_{\run}+1}}
	&= \exof{ \braket{\noise_{\head_{\run}+1}}{\state_{\run} - \base} }\
	\notag\\
	&= \exof{ \exof{\braket{\noise_{\head_{\run}+1}}{\state_{\run} - \base} \given \filter_{\run}} }
	= \exof{ \braket{\exof{\noise_{\head_{\run}+1} \given \filter_{\run}}}{\state_{\run} - \base} }
	= 0
\end{align}
where the last step follows from the fact that $\exof{\noise_{\head_{\run}+1} \given \filter_{\run}} = \exof{\est\payv_{\run} - \exof{\est\payv_{\run} \given \filter_{\run}} \given \filter_{\run}} = 0$.
\end{proof}

\begin{lemma}
\label{lem:sbias}
Let $\sbias_{\head_{\run}} = \braket{\bias_{\head_{\run}}}{\state_{\run} - \base}$.
Then, there exists a positive constant $\bbound>0$ such that:
\begin{equation}
\label{eq:sbias}
\abs{\sbias_{\head_{\run}}}
	\leq \bbound \mix_{\head_{\run}}
	\quad
	\text{for all $\run=\running$}
\end{equation}
\end{lemma}

\begin{proof}
In this result, the role of the choices of other players plays an important role, so we will momentarily reinstate the player superscript.
However, to keep the notation manageable, we will tacitly assume that all players trigger the same delay process so $\head_{\run}^{\play} = \head_{\run}^{\playalt}$ for all $\play,\playalt\in\players$;
the proof is exactly the same in the general case.

To begin, by the definition \eqref{eq:oracle} of $\est\payv_{\run}$ and the independence of the sampling directions $\unitvar_{\runalt}^{\play}$ across players $\play\in\players$ and stages $\runalt=\running\run$, we have
\begin{flalign}
\exof{\est\payv_{\run}^{\play} \given \filter_{\run}}
	&= \frac{\vdim^{\play}/\mix_{\head_{\run}}}{\prod_{\playalt} \vol(\sphere^{\playalt})}
	\int_{\sphere^{1}} \dotsi \int_{\sphere^{\nPlayers}}
		\pay^{\play}(
			\state_{\head_{\run}}^{1} + \mix_{\head_{\run}}\unitvec^{1},
			\dotsc,
			\state_{\head_{\run}}^{\nPlayers} + \mix_{\head_{\run}}\unitvec^{\nPlayers})
		\unitvec^{\play}\,
		\dd \unitvec^{1}\dotsm \dd \unitvec^{\nPlayers}
	\notag\\
	&= \frac{\vdim^{\play}/\mix_{\head_{\run}}}{\prod_{\playalt} \vol(\mix_{\head_{\run}}\sphere^{\playalt})}
	\int_{\mix_{\head_{\run}}\sphere^{1}} \dotsi \int_{\mix_{\head_{\run}}\sphere^{\nPlayers}}
		\pay^{\play}(\state_{\head_{\run}}^{1} + \unitvec^{1},\dotsc,\state_{\head_{\run}}^{\nPlayers} + \unitvec^{\nPlayers})
		\frac{\unitvec^{\play}}{\norm{\unitvec^{\play}}}\,
		\dd \unitvec^{1}\dotsm \dd \unitvec^{\nPlayers}
	\notag\\
	&= \frac{\vdim^{\play}/\mix_{\head_{\run}}}{\prod_{\playalt} \vol(\mix_{\head_{\run}}\sphere^{\playalt})}
	\int_{\mix_{\head_{\run}}\sphere^{\play}} \int_{\prod_{\playalt\neq\play} \mix_{\head_{\run}}\sphere^{\playalt}}
		\pay^{\play}(\state_{\head_{\run}}^{\play} + \unitvec^{\play};\state_{\head_{\run}}^{-\play} + \unitvec^{-\play})
		\frac{\unitvec^{\play}}{\norm{\unitvec^{\play}}}\,
		\dd \unitvec^{\play} \dd \unitvec^{-\play}
	\notag\\
	&= \frac{\vdim^{\play}/\mix_{\head_{\run}}}{\prod_{\playalt} \vol(\mix_{\head_{\run}}\sphere^{\playalt})}
	\int_{\mix_{\head_{\run}}\ball^{\play}} \int_{\prod_{\playalt\neq\play} \mix_{\head_{\run}}\sphere^{\playalt}}
		\nabla_{\play}\pay^{\play}(\state_{\head_{\run}}^{\play} + \pertvec^{\play};\state_{\head_{\run}}^{-\play} + \unitvec^{-\play})
		\dd \pertvec^{\play} \dd \unitvec^{-\play},
\end{flalign}
where, in the last line, we used Stokes' theorem \citep{Lee03} to write
\begin{equation}
\label{eq:Stokes}
\int_{\mix_{\head_{\run}}\sphere^{\play}}
		\pay^{\play}(\state_{\head_{\run}}^{\play} + \unitvec^{\play};\state_{\head_{\run}}^{-\play} + \unitvec^{-\play})
		\frac{\unitvec^{\play}}{\norm{\unitvec^{\play}}}\,
	\dd \unitvec^{\play}
	= \int_{\mix_{\head_{\run}}\ball^{\play}}
	\nabla_{\play} \pay^{\play}(\state_{\head_{\run}}^{\play} + \pertvec^{\play};\state_{\head_{\run}}^{-\play} + \unitvec^{-\play})
	\dd \pertvec^{\play}.
\end{equation}
Since $\vol(\mix_{\head_{\run}} \ball^{\play}) = (\mix_{\head_{\run}}/\vdim^{\play}) \vol(\mix_{\head_{\run}}\sphere^{\play})$, the above yields
\begin{equation}
\exof{\est\payv_{\run}^{\play} \given \filter_{\run}}
	= \nabla_{\play} \pay_{\mix_{\head_{\run}}}^{\play}(\state_{\head_{\run}}^{\play};\state_{\head_{\run}}^{-\play})
\end{equation}
where the ``$\mix$-smoothed'' payoff function $\pay_{\mix}^{\play}$ of player $\play$ is defined as
\begin{equation}
\label{eq:pay-smooth}
\pay_{\mix}^{\play}(\point^{\play};\point^{-\play})
	= \frac{1}{\vol(\mix\ball^{\play}) \prod_{\playalt\neq\play} \vol(\mix\sphere_{\playalt})}
	\int_{\mix\ball^{\play}} \int_{\prod_{\playalt\neq\play} \mix\sphere^{\playalt}}
		\pay^{\play}(\point^{\play} + \pertvec^{\play};\point^{-\play} + \unitvec^{-\play})\,
		\dd\unitvec^{1} \dotsm d\pertvec^{\play} \dotsm d\unitvec^{\nPlayers}
\end{equation}

We now proceed to show that $\max_{\point\in\points} \norm{\nabla_{\play} \pay_{\mix}^{\play}(\point) - \nabla_{\play}\pay^{\play}(\point)} = \bigoh(\mix)$.
Indeed, by \cref{asm:pay}, we have $\norm{\payv^{\play}(\pointalt) - \payv^{\play}(\pointalt)} \leq \smooth \norm{\pointalt - \point}$ for all $\point,\pointalt\in\points$.
Hence, for $\pertvec^{\play}\in\mix\ball^{\play}$ and all $\unitvec^{\playalt} \in \mix\sphere^{\playalt}$, $\playalt\neq\play$, we obtain:
\begin{equation}
\norm{\nabla_{\play} \pay^{\play}(\point^{\play} + \pertvec^{\play};\point^{-\play} + \unitvec^{-\play}) - \nabla_{\play}\pay^{\play}(\point^{\play};\point^{-\play})}
	\leq \smooth \sqrt{\norm{\pertvec^{\play}}^{2} + \sum\nolimits_{\playalt\neq\play} \norm{\unitvec^{\playalt}}^{2}}
	\leq \smooth \mix \sqrt{\nPlayers}.
\end{equation}
Thus, by differentiating under the integral sign in the definition \eqref{eq:pay-smooth} of $\pay_{\mix}^{\play}$, we get:
\begin{align}
\norm{\nabla_{\play} \pay_{\mix}^{\play}(\point) - \nabla_{\play}\pay^{\play}(\point)}
	&= \frac{1}{\vol(\mix\ball^{\play}) \prod_{\playalt\neq\play} \vol(\mix\sphere_{\playalt})}
	\notag\\
	&\times
		\norm*{
			\int_{\mix\ball^{\play}} \int_{\prod_{\playalt\neq\play} \mix\sphere^{\playalt}}
			\nabla_{\play} \pay^{\play}(\point^{\play} + \pertvec^{\play};\point^{-\play} + \unitvec^{-\play})
			- \nabla_{\play}\pay^{\play}(\point^{\play};\point^{-\play})
			\,\dd\unitvec^{1} \dotsm d\pertvec^{\play} \dotsm d\unitvec^{\nPlayers}
			} 
	\notag\\
	&\leq \frac{1}{\vol(\mix\ball^{\play}) \prod_{\playalt\neq\play} \vol(\mix\sphere_{\playalt})}
	\notag\\
	&\times
		\int_{\mix\ball^{\play}} \int_{\prod_{\playalt\neq\play} \mix\sphere^{\playalt}}
		\norm*{
			\nabla_{\play} \pay^{\play}(\point^{\play} + \pertvec^{\play};\point^{-\play} + \unitvec^{-\play})
			- \nabla_{\play}\pay^{\play}(\point^{\play};\point^{-\play})
			} 
			\,\dd\unitvec^{1} \dotsm d\pertvec^{\play} \dotsm d\unitvec^{\nPlayers}
	\notag\\
	&\leq \sqrt{\nPlayers} \smooth \mix.
\end{align}
With all this in hand, we finally get:
\begin{align}
\abs{\sbias_{\head_{\run}}^{\play}}
	&\leq \norm{\bias_{\head_{\run}}^{\play}} \norm{\state_{\run}^{\play} - \base^{\play}}
	\notag\\
	&\leq \diam(\points^{\play}) \norm{\exof{\est\payv_{\run}^{\play} \given \filter_{\run}} - \payv(\state_{\head_{\run}}^{\play})}
	= \diam(\points^{\play}) \norm{\nabla_{\play} \pay_{\mix_{\head_{\run}}}^{\play}(\state_{\head_{\run}}^{\play}) - \nabla_{\play}\pay^{\play}(\state_{\head_{\run}})}
	\notag\\
	&\leq \diam(\points^{\play}) \sqrt{\nPlayers} \smooth \mix_{\head_{\run}},
\end{align}
and our proof is complete.
\end{proof}

Finally, for the last term in the template inequality \eqref{eq:template}, we have:

\begin{lemma}
\label{lem:svar}
Let $\svar_{\run}^{2} = \frac{1}{2} \norm{\est\payv_{\run}}^{2}$.
Then, there exists a positive constant $\vbound>0$ such that:
\begin{equation}
\label{eq:svar}
\exof{\svar_{\run}^{2} \given \filter_{\run}}
	\leq \frac{\vbound^{2}}{2\mix_{\head_{\run}}^{2}}
	\quad
	\text{for all $\run=\running$}
\end{equation}
\end{lemma}

\begin{proof}
By the definition \eqref{eq:oracle} of $\est\payv_{\run}$, we have:
\begin{equation}
\exof{\svar_{\run}^{2} \given \filter_{\run}}
	= \frac{1}{2} \frac{\vdim^{2}}{\mix_{\head_{\run}}^{2}} \exof{\est\pay_{\head_{\run}}^{2} \norm{\unitvar_{\head_{\run}}}^{2} \given \filter_{\run}}
	\leq \frac{\vbound^{2}}{2\mix_{\head_{\run}}^{2}}
\end{equation}
where $\vbound^{2} = \vdim^{2} \max_{\point\in\points} \abs{\pay(\point)}$ and we used the fact that $\norm{\unitvar_{\run}} = 1$ (by construction).
\end{proof}

\section{Properties of the pooling strategy}
\label{app:pooling}

In this section we prove some key properties of the information queueing/dequeing strategy that is employed by the \ac{GOLD} algorithm.
This strategy allows the players to receive rewards and thus to update their strategies reasonably often (\cf \cref{lem:noupdate}), while ensuring that rewards do not stay in the unused pool any longer than absolutely necessary (\cref{lem:delayedfeedback}).
We make these statements precise belo:

\begin{lemma}
\label{lem:noupdate}
The maximal number of steps with an empty pool \textpar{no update} is bounded by the maximal delay encountered up to the stage in question;
formally:
\begin{equation}
\#\setdef{\runalt}{\pool_\runalt=\varnothing, \runalt=\running\run}
	\leq \max\nolimits_{1\leq \runalt \leq\run} (\delay_\runalt)
	\quad
	\text{for all $\run=\running$}
\end{equation}
\end{lemma}

\begin{proof}
We begin by analyzing the \emph{constant delay} case, \ie when $\delay_{\run}=\Delay$ for all $\runalt=\running\run$.
In that case, the first $\Delay$ steps have no update;
subsequently, at each stage, the agent receives payoff information with delay $\Delay$.

Now assume that $\delay_{\runalt}$, $\runalt=\running\run$, is a given sequence of delays.
We will show that if we modify a term of this sequence to $\alt\delay_{\runalt} < \delay_{\runalt}$, the number of steps without an update will remain the same or decrease by one.
Indeed, given that the reward of the $\runalt$-th stage is now collected at stage $\alt\delay_{\runalt} + \runalt$ instead of the (later) stage $\delay_{\runalt} + \runalt$, the size of the pool $\pool_{\alt\delay_{\runalt} + \runalt}$ will be increased by $1$.
In turn, this provides an update to the next step without an update (which could be $\alt\delay_{\runalt} +  \runalt$).
Thus, if the next step without an update was before time $\run$, the number of steps without an update decreases by one, otherwise it remains unchanged.
This means that the maximal number of steps without an update is reached when the delay is constant and is equal to said delay.
Our claim then follows immediately (see also \cref{fig:delay-change} for a graphical illustration of this argument).
\end{proof}


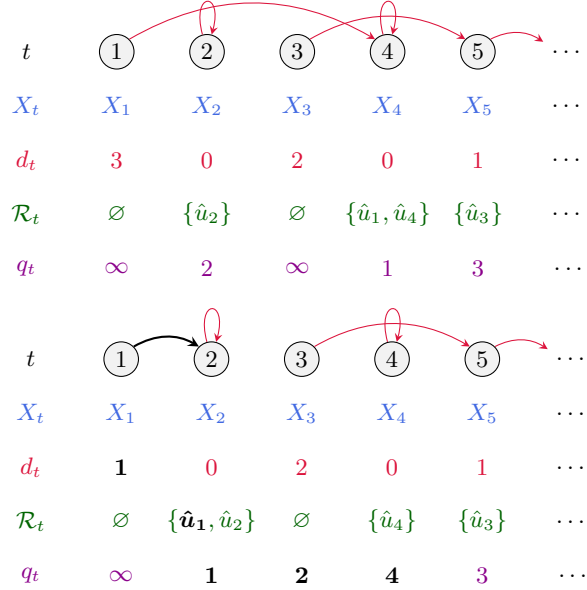
\begin{figure}[t]
\centering

\begin{tikzpicture}
[scale=1.2,
nodestyle/.style={circle,draw=black,fill=gray!10, inner sep=2pt},
edgestyle/.style={-},
>=stealth]

\small

\def\dx{1}
\def\dy{.6}

\coordinate (T0) at (0,0);
\coordinate (X0) at (0,-\dy);
\coordinate (D0) at (0,-2*\dy);
\coordinate (F0) at (0,-3*\dy);
\coordinate (Q0) at (0,-4*\dy);

\coordinate (T1) at (1,0);
\coordinate (X1) at (1,-\dy);
\coordinate (D1) at (1,-2*\dy);
\coordinate (F1) at (1,-3*\dy);
\coordinate (Q1) at (1,-4*\dy);

\coordinate (T2) at (2,0);
\coordinate (X2) at (2,-\dy);
\coordinate (D2) at (2,-2*\dy);
\coordinate (F2) at (2,-3*\dy);
\coordinate (Q2) at (2,-4*\dy);

\coordinate (T3) at (3,0);
\coordinate (X3) at (3,-\dy);
\coordinate (D3) at (3,-2*\dy);
\coordinate (F3) at (3,-3*\dy);
\coordinate (Q3) at (3,-4*\dy);

\coordinate (T4) at (4,0);
\coordinate (X4) at (4,-\dy);
\coordinate (D4) at (4,-2*\dy);
\coordinate (F4) at (4,-3*\dy);
\coordinate (Q4) at (4,-4*\dy);

\coordinate (T5) at (5,0);
\coordinate (X5) at (5,-\dy);
\coordinate (D5) at (5,-2*\dy);
\coordinate (F5) at (5,-3*\dy);
\coordinate (Q5) at (5,-4*\dy);

\coordinate (T6) at (6,0);
\coordinate (X6) at (6,-\dy);
\coordinate (D6) at (6,-2*\dy);
\coordinate (F6) at (6,-3*\dy);
\coordinate (Q6) at (6,-4*\dy);

\coordinate (T7) at (7,0);
\coordinate (X7) at (7,-\dy);
\coordinate (D7) at (7,-2*\dy);
\coordinate (F7) at (7,-3*\dy);
\coordinate (Q7) at (7,-4*\dy);


\coordinate (Tlast) at (6,0);
\coordinate (Xlast) at (6,-\dy);
\coordinate (Dlast) at (6,-2*\dy);
\coordinate (Flast) at (6,-3*\dy);
\coordinate (Qlast) at (6,-4*\dy);

\node (T0) at (T0) {$\run$};
\node (X0) at (X0) [text = RoyalBlue] {$\state_\run$};
\node (D0) at (D0) [text = Crimson] {$\delay_\run$};
\node (F0) at (F0) [text = DarkGreen] {$\rewards_\run$};
\node (Q0) at (Q0) [text = DarkMagenta] {$\head_\run$};


\node (T1) at (T1) [nodestyle] {$1$};
\node (X1) at (X1) [text = RoyalBlue] {$\state_{1}$};
\node (D1) at (D1) [text = Crimson] {$3$};
\node (F1) at (F1) [text = DarkGreen] {$\varnothing$};
\node (Q1) at (Q1) [text = DarkMagenta] {$\infty$};

\node (T2) at (T2) [nodestyle] {$2$};
\node (X2) at (X2) [text = RoyalBlue] {$\state_{2}$};
\node (D2) at (D2) [text = Crimson] {$0$};
\node (F2) at (F2) [text = DarkGreen] {$\{\est\pay_{2}\}$};
\node (Q2) at (Q2) [text = DarkMagenta] {$2$};
\draw [Crimson,edgestyle,->,loop above] (T2) to (T2);

\node (T3) at (T3) [nodestyle] {$3$};
\node (X3) at (X3) [text = RoyalBlue] {$\state_{3}$};
\node (D3) at (D3) [text = Crimson] {$2$};
\node (F3) at (F3) [text = DarkGreen] {$\varnothing$};
\node (Q3) at (Q3) [text = DarkMagenta] {$\infty$};

\node (T4) at (T4) [nodestyle] {$4$};
\node (X4) at (X4) [text = RoyalBlue] {$\state_{4}$};
\node (D4) at (D4) [text = Crimson] {$0$};
\node (F4) at (F4) [text = DarkGreen]  {$\{\est\pay_{1},\est\pay_{4}\}$};
\node (Q4) at (Q4) [text = DarkMagenta] {$1$};
\draw [Crimson,edgestyle,->,bend left] (T1.north east) to (T4.north west);
\draw [Crimson,edgestyle,->,loop above] (T4) to (T4);

\node (T5) at (T5) [nodestyle] {$5$};
\node (X5) at (X5) [text = RoyalBlue] {$\state_{5}$};
\node (D5) at (D5) [text = Crimson] {$1$};
\node (F5) at (F5) [text = DarkGreen] {$\{\est\pay_{3}\}$};
\node (Q5) at (Q5) [text = DarkMagenta] {$3$};
\draw [Crimson,edgestyle,->,bend left] (T3.north east) to (T5.north west);



\node (Tlast) at (Tlast) {$\dots$};
\node (Xlast) at (Xlast) {$\dots$};
\node (Dlast) at (Dlast) {$\dots$};
\node (Flast) at (Flast) {$\dots$};
\node (Qlast) at (Qlast){$\ldots$};
\draw [Crimson,edgestyle,->,bend left] (T5.north east) to (Tlast.north west);

%
%

\end{tikzpicture}
\medskip

\begin{tikzpicture}
[scale=1.2,
nodestyle/.style={circle,draw=black,fill=gray!10, inner sep=2pt},
edgestyle/.style={-},
>=stealth]

\small

\def\dx{1}
\def\dy{.6}

\coordinate (T0) at (0,0);
\coordinate (X0) at (0,-\dy);
\coordinate (D0) at (0,-2*\dy);
\coordinate (F0) at (0,-3*\dy);
\coordinate (Q0) at (0,-4*\dy);

\coordinate (T1) at (1,0);
\coordinate (X1) at (1,-\dy);
\coordinate (D1) at (1,-2*\dy);
\coordinate (F1) at (1,-3*\dy);
\coordinate (Q1) at (1,-4*\dy);

\coordinate (T2) at (2,0);
\coordinate (X2) at (2,-\dy);
\coordinate (D2) at (2,-2*\dy);
\coordinate (F2) at (2,-3*\dy);
\coordinate (Q2) at (2,-4*\dy);

\coordinate (T3) at (3,0);
\coordinate (X3) at (3,-\dy);
\coordinate (D3) at (3,-2*\dy);
\coordinate (F3) at (3,-3*\dy);
\coordinate (Q3) at (3,-4*\dy);

\coordinate (T4) at (4,0);
\coordinate (X4) at (4,-\dy);
\coordinate (D4) at (4,-2*\dy);
\coordinate (F4) at (4,-3*\dy);
\coordinate (Q4) at (4,-4*\dy);

\coordinate (T5) at (5,0);
\coordinate (X5) at (5,-\dy);
\coordinate (D5) at (5,-2*\dy);
\coordinate (F5) at (5,-3*\dy);
\coordinate (Q5) at (5,-4*\dy);

\coordinate (T6) at (6,0);
\coordinate (X6) at (6,-\dy);
\coordinate (D6) at (6,-2*\dy);
\coordinate (F6) at (6,-3*\dy);
\coordinate (Q6) at (6,-4*\dy);

\coordinate (T7) at (7,0);
\coordinate (X7) at (7,-\dy);
\coordinate (D7) at (7,-2*\dy);
\coordinate (F7) at (7,-3*\dy);
\coordinate (Q7) at (7,-4*\dy);


\coordinate (Tlast) at (6,0);
\coordinate (Xlast) at (6,-\dy);
\coordinate (Dlast) at (6,-2*\dy);
\coordinate (Flast) at (6,-3*\dy);
\coordinate (Qlast) at (6,-4*\dy);

\node (T0) at (T0) {$\run$};
\node (X0) at (X0) [text = RoyalBlue] {$\state_\run$};
\node (D0) at (D0) [text = Crimson] {$\delay_\run$};
\node (F0) at (F0) [text = DarkGreen] {$\rewards_\run$};
\node (Q0) at (Q0) [text = DarkMagenta] {$\head_\run$};


\node (T1) at (T1) [nodestyle] {$1$};
\node (X1) at (X1) [text = RoyalBlue] {$\state_{1}$};
\node (D1) at (D1)  {$\mathbf{1}$};
\node (F1) at (F1) [text = DarkGreen] {$\varnothing$};
\node (Q1) at (Q1) [text = DarkMagenta] {$\infty$};

\node (T2) at (T2) [nodestyle] {$2$};
\node (X2) at (X2) [text = RoyalBlue] {$\state_{2}$};
\node (D2) at (D2) [text = Crimson] {$0$};
\node (F2) at (F2) [text = DarkGreen] {$\{\color{black}\boldsymbol{\est\pay_{1}}\color{DarkGreen},\est\pay_{2}\}$};
\node (Q2) at (Q2){$\mathbf{1}$};
\draw [Crimson,edgestyle,->,loop above] (T2) to (T2);
\draw [thick,edgestyle,->,bend left] (T1.north east) to (T2.north west);

\node (T3) at (T3) [nodestyle] {$3$};
\node (X3) at (X3) [text = RoyalBlue] {$\state_{3}$};
\node (D3) at (D3) [text = Crimson] {$2$};
\node (F3) at (F3) [text = DarkGreen] {$\varnothing$};
\node (Q3) at (Q3) {$\mathbf{2}$};

\node (T4) at (T4) [nodestyle] {$4$};
\node (X4) at (X4) [text = RoyalBlue] {$\state_{4}$};
\node (D4) at (D4) [text = Crimson] {$0$};
\node (F4) at (F4) [text = DarkGreen]  {$\{\est\pay_{4}\}$};
\node (Q4) at (Q4) {$\mathbf{4}$};
\draw [Crimson,edgestyle,->,loop above] (T4) to (T4);

\node (T5) at (T5) [nodestyle] {$5$};
\node (X5) at (X5) [text = RoyalBlue] {$\state_{5}$};
\node (D5) at (D5) [text = Crimson] {$1$};
\node (F5) at (F5) [text = DarkGreen] {$\{\est\pay_{3}\}$};
\node (Q5) at (Q5) [text = DarkMagenta] {$3$};
\draw [Crimson,edgestyle,->,bend left] (T3.north east) to (T5.north west);



\node (Tlast) at (Tlast) {$\dots$};
\node (Xlast) at (Xlast) {$\dots$};
\node (Dlast) at (Dlast) {$\dots$};
\node (Flast) at (Flast) {$\dots$};
\node (Qlast) at (Qlast){$\ldots$};
\draw [Crimson,edgestyle,->,bend left] (T5.north east) to (Tlast.north west);

%
%

\end{tikzpicture}
\caption{Evolution of the number of steps without an update if one of the delays encountered is reduced.
The case presented is when the number of steps without an update decreases.
Here the delay of step $1$ decreases from $3$ to $1$, and, as a consequence, step $3$ has an update that it previously did not have (changes in the information pool and its head are higlighted in \textbf{bold}).}
\label{fig:delay-change}
\end{figure}

We now proceed to bound the difference between the time that an information is used and the time it was generated:

\begin{lemma}
\label{lem:delayedfeedback}
The time lag $\run-\head_{\run}$ between the step when a reward is observed and the step when the action was taken is bounded from above by the maximal delay up to step $\run$;
formally:
\begin{equation}
\label{ex:lag}
\run - \head_{\run}
	\leq \max\nolimits_{\start \leq \runalt \leq \run} \delay_{\runalt}
	\quad
	\text{for all $\run=\running$}
\end{equation}
\end{lemma}

\begin{proof}
We proceed by induction.
Since at first the pool is empty, one of the rewards arriving in the first non-empty $\info$ is used, so $\run-\head_{\run}=\delay_{\head_{\run}}\leq \max\limits_{\runalt \leq \run} (\delay_{\runalt})$ in this case.

For the inductive step, assume that $\run - \head_{\run} \leq \max\limits_{\runalt \leq \run} (\delay_{\runalt})$ for some $\run\geq1$;
we prove below that this is also the case when $\run\leftarrow\run +1$.
In that case, $\pool_{\run+1} \leftarrow \pool_{\run} \cup \info_{\run+1}$ and $\head_{\run+1} = \min(\pool_{\run+1})$ so $ \min(\info_{\run+1})\geq \run +1 - \max\limits_{\runalt \leq \run +1} (\delay_{\runalt})$.
In addition, at the $\run$-th step, the oldest element of $\pool_{\run} $ is removed from the pool and utilized.
As a result, we have
\begin{equation}
\min \pool_{\run}
	> \run - \max\nolimits_{\start \leq \runalt \leq \run} \delay_{\runalt}
	\geq \run +1 - \max\nolimits_{\start\leq\runalt \leq \run +1} \delay_{\runalt}
\end{equation}
This gives $\head_{\run+1} \geq \run +1 - \max_{\start \leq \runalt \leq \run +1} \delay_{\runalt}$ and completes the induction (and our proof).
\end{proof}

Our next lemma goes a step further and shows that, under our blanket assumptions, the lag $\run-\head_{\run}$ is small relative to $\run$:

\begin{corollary}
\label{cor:lag}
Under \cref{asm:delay}, we have $\run - \head_{\run} = o(\run^{\alpha})$;
in particular, $\run-\head_{\run} = o(\run)$ and $\head_{\run}=\Theta(\run)$ .
\end{corollary}

\begin{proof}
This follows immediately from \cref{lem:delayedfeedback,asm:delay}.
\end{proof}

We proceed with the proof of \cref{lem:3sums} which establishes the required control on the sequences
\begin{equation}
\tag{\ref{eq:3seqs}}
A_{\run}
	= \step_{\run} \sum_{\runalt=\head_{\run}}^{\run-1} \frac{\step_{\runalt}}{\mix_{\runalt}},
	\quad
B_{\run}
	= \step_{\run} \mix_{\head_{\run}},
	\quad
\textrm{and}
	\quad
C_{\run}
	= \frac{\step_\run^{2}}{\mix_{\head_{\run}}^{2}},
\end{equation}
that couple the encountered delays with the step-size and sampling policies of \eqref{eq:GOLD}.
For convenience, we restate \cref{lem:3sums} below:

\threesums*

\begin{remark*}
For clarity, the conditions on $b$, $c$ and $\alpha$ are summarized in \cref{fig:b_and_c} below.
\end{remark*}

\begin{figure}[t]
\centering

\begin{tikzpicture}
[scale=0.5,
nodestyle/.style={circle,draw=black,fill=gray!10, inner sep=2pt},
edgestyle/.style={-},
>=stealth]

\small

\def\dx{1}
\def\dy{.6}

\draw[thick,->] (0,-.5) -- (0,9.5) node[anchor=north east] {b};
\draw[thick,->](-.5,0) -- (9.5,0) node[anchor=north west] {c};

\fill[DarkGreen!40!white] (8,0) --(8,4)--(6,2);

\draw[thick,dashed,RoyalBlue] (9,-1) -- (-1,9);
\draw[thick,RoyalBlue] (8,-1) -- (8,9) ;
\draw[thick,dashed,RoyalBlue] (3,.-1) -- (9,5);

\draw[thick,Crimson] (7,-2) -- (9, 2);
\coordinate (A0) at (9.75, 2.5);
\node (A0) at (A0) [text = Crimson] {$\alpha = 1$};
\draw[thick,Crimson] (5,-2) -- (9, 6);
\coordinate (A1) at (9.75, 6.5);
\node (A1) at (A1) [text = Crimson] {$\alpha =  \frac{1}{2}$};
\draw[thick,Crimson] (4,-2) -- (9, 8);
\coordinate (A2) at (9.75, 8.5);
\node (A2) at (A2) [text = Crimson] {$\alpha =  \frac{1}{4}$};
\draw[thick,Crimson] (3,-2) -- (9,10);
\coordinate (A3) at (9.75, 10.5);
\node (A3) at (A3) [text = Crimson] {$\alpha = 0$};

\draw[step=1cm,dash dot,gray,very thin] (-.2,-.2) grid (8.7,8.7);

\coordinate (A4) at (6.5, 7);
\node (A4) at (A4) [text = Crimson,rotate=63.4349488] {$b \leq 2c - 1 - \alpha$};

\draw (6,0.2) -- (6,-0.2) node[anchor=north] {$\frac{3}{4}$};
\draw ( 8,0.2) -- (8,-0.2) node[anchor=north west] {$1$};
\draw (0.2,8) -- (-0.2,8) node[anchor=east] {$1$};

\end{tikzpicture}
\caption{The allowable region (green shaded areas) of possible values of the sampling radius and step-size exponents $b$ and $c$ for various values of the groth exponent $\alpha$ of the encountered delays.
The dashed blue lines corresponding to the last two terms in \eqref{eq:params} indicate hard boundaries leading to logarithmic terms in the regret instead of constants.}
\label{fig:b_and_c}
\end{figure}
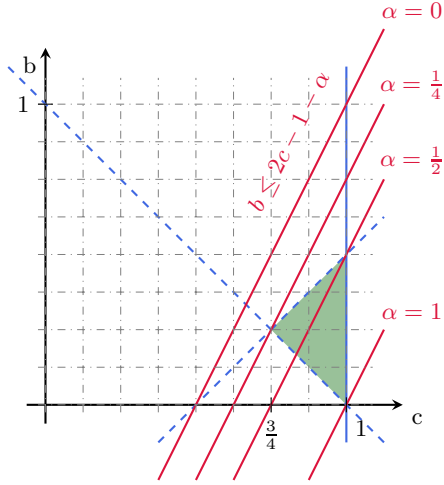

\begin{proof}
We proceed step-by-step.
\begin{enumerate}
\item
For the series $\sum_{\run=\start}^{\nRuns} A_{\run}$, \cref{cor:lag} guarantees the existence of some $M$ such that $\run-\head_{\run}\leq M \run^{\alpha} \leq M \run$.
We also have that $1/\runalt^{b} \propto \mix_{\runalt} \leq \mix_{\head_{\runalt}} \propto 1/\head_{\runalt}^{b} = \Theta(1/\runalt^{b})$ so $\mix_{\head_{\runalt}} = \Theta(1/\runalt^{b})$.
In turn, this gives
\begin{align}
A_{\run}
	= \step_{\run} \sum_{\runalt=\head_{\run}}^{\run-1} \frac{\step_{\runalt}}{\mix_{\runalt}}
	\leq \step_{\run} (\run - \head_{\run}) \frac{\step_{\run}}{\mix_{\head_{\run}}}
	\leq M \step_{\run} \run^{\alpha} \frac{\step_{\run}}{\mix_{\head_{\run}}}
	= \Theta\parens*{\frac{1}{\run^{2c-\alpha-b}}}.
\end{align}
We conclude that $\sum_{\run=\start}^{\infty} A_{\run}$ is finite if $2c-\alpha-b>1$ and $\sum_{\run=\start}^{\nRuns} A_{\run} = \bigoh(\log\nRuns)$ if $2c - \alpha - b =1$;
this establishes our claim for $A_{\run}$.

\item
For the series $\sum_{\run=\start}^{\nRuns} B_{\run}$, invoking again \cref{cor:lag} and arguing as above, we readily get $\step_{\run} \mix_{\head_{\run}} = \bigoh(1/\run^{b+c})$.
Hence, the condition $c + b > 1$ is sufficient for the convergence of the infinite series to a finite number, whereas, in the case $c + b = 1$, we have $\sum_{\run=\start}^{\nRuns} B_{t} = \bigoh(\log\nRuns)$.

\item
Finally, for the series $\sum_{\run=\start}^{\nRuns} C_{\run} = \sum_{\run=\start}^{\nRuns} \step_{\run}^{2} / \mix_{\head_{\run}}^{2}$, invoking \cref{cor:lag} one last time and using the fact that $\mix_{\head_{\run}} = \bigoh(1/\run^{b})$, we obtain $\step_{\run}^{2} / \mix_{\head_{\run}}^{2} = \Theta(\run^{2b} / \run^{2c})$.
Therefore, the condition $2c - 2b>1$ is sufficient for the convergence of the infinite series to a finite number, whereas, in the case $2c - 2b = 1$, we have $\sum_{\run=\start}^{\nRuns} C_{\run} = \bigoh(\log\nRuns)$.
\qedhere
\end{enumerate}
\end{proof}

\section{Regret analysis}
\label{app:regret}

We are now in a position to prove the main regret guarantee of the \ac{GOLD} policy;
for convenience, we restate it below:

\regret*

\begin{proof}
Fix a benchmark action $\point\in\points$.
Then, for all stages $\run = \running\nRuns$ at which an update occurs, \cref{lem:template} gives
\begin{align}
\step_{\run} \braket{\payv(\state_{\run})}{\point -\state_{\run}} 
	&\leq \frac{1}{2} \norm{\state_{\run} - \point}^{2}
		- \frac{1}{2}\norm{\point-\state_{\run+1}}^{2}
	\notag\\
	&+ \step_{\run} \sum_{\runalt=\head_{\run}}^{\run-1} \braket{\payv(\state_{\runalt}) - \payv(\state_{\runalt +1})}{\state_{\run} - \point}
	\notag\\
	&+ \step_{\run}\snoise_{\head_{\run}+1}
		+ \step_{\run}\sbias_{\head_{\run}}
		+ \step_{\run}^{2} \svar_{\run}^{2}
\end{align}
where, the quantities $\snoise_{\head_{\run}+1}$, $\sbias_{\head_{\run}}$ and $\svar_{\run}^{2}$ are defined as in \cref{lem:snoise,lem:sbias,lem:svar} respectively,
and, in the second line, we unfolded the pairing $\braket{\payv(\state_{\head_{\run}})}{\state_{\run} - \point}$ as
\begin{align}
\label{eq:unfold}
\braket{\payv(\state_{\head_{\run}})}{\state_{\run} - \point}
	&= \braket{\payv(\state_{\head_{\run}}) - \payv(\state_{\head_{\run}+1})}{\state_{\run} - \point}
	+ \dotsm
	+ \braket{\payv(\state_{\run})}{\state_{\run} - \point}
	\notag\\
	&= \braket{\payv(\state_{\run})}{\state_{\run} - \point}
	+ \sum_{\runalt=\head_{\run}}^{\run-1} \braket{\payv(\state_{\runalt}) - \payv(\state_{\runalt +1})}{\state_{\run} - \point}.
\end{align}
Therefore, conditioning on $\filter_{\run}$ and taking expectations, we get the bound:
\begin{subequations}
\label{eq:temp-bound}
\begin{align}
\step_{\run} \braket{\payv(\state_{\run})}{\point -\state_{\run}}
	&= \step_{\run} \exof{\braket{\payv(\state_{\run})}{\point -\state_{\run}} \given \filter_{\run}}
	\notag\\
	&\leq \frac{1}{2} \norm{\point-\state_{\run}}^{2}
		- \frac{1}{2} \exof{\norm{\point-\state_{\run+1}}^{2} \given \filter_{\run}}
	\\
	&\label{eq:reg-bound-unfold}
	+ \step_{\run} \exof*{
		\sum_{\runalt=\head_{\run}}^{\run-1}
		\braket{\payv(\state_{\runalt}) - \payv(\state_{\runalt +1})}{\state_{\run} - \point}
			\given \filter_{\run}}
	\\
	&\label{eq:reg-bound-lems}
	+ \step_{\run} \exof{\snoise_{\head_{\run}+1}\given \filter_{\run}}
		+ \step_{\run} \exof{\sbias_{\head_{\run}} \given \filter_{\run}}
		+ \step _{\run}^{2} \exof{\svar_{\run}^{2} \given \filter_{\run}},
\end{align}
\end{subequations}
where, in the first line, we used the fact that $\state_{\run}$ is $\filter_{\run}$-measurable.

We proceed term-by-term.
First, for \eqref{eq:reg-bound-unfold}, since $\state_{\runalt}$ is $\filter_{\run}$-measurable for $\runalt\leq \run$, we get:
\usetagform{comment}
\begin{align}
\label{eq:unfold-bound}
\exof*{\sum_{\runalt=\head_{\run}}^{\run-1}
	\braket{\payv(\state_{\runalt}) - \payv(\state_{\runalt +1})}{\state_{\run} - \point}
	\given \filter_{\run}}
	\hspace{-12em}&
	\notag\\
	&= \sum_{\runalt=\head_{\run}}^{\run-1}
	\braket{\payv(\state_{\runalt}) - \payv(\state_{\runalt +1})}{\state_{\run} - \point}
	\tag{measurability}\\
	&\leq \sum_{\runalt=\head_{\run}}^{\run-1}
		\norm{\payv(\state_{\runalt}) - \payv(\state_{\runalt +1})}
		\cdot \norm{\state_{\run} - \point}
	\tag{Cauchy-Schwarz}
	\\
	&\leq \smooth \diam(\points)
		\sum_{\runalt=\head_{\run}}^{\run-1}
			\norm{\state_{\runalt} - \state_{\runalt +1}}
	\tag{Lipschitz $+$ compactness}
	\\
	&\leq \smooth \diam(\points)
		\sum_{\runalt=\head_{\run}}^{\run-1}
			\norm{\step_{\runalt} \est\payv_{\runalt}}
	\tag{non-expansivity of $\Eucl$}
	\\
	&= \smooth \diam(\points)
		\sum_{\runalt=\head_{\run}}^{\run-1}
			\step_{\runalt} \frac{\vdim}{\mix_{\head_{\runalt}}} \est\pay_{\head_{\runalt}}
	\tag{\ac{SPSA} estimator}
	\\
	&\leq K \sum_{\runalt=\head_{\run}}^{\run-1} \frac{\step_{\runalt}}{\mix_{\runalt}}
	\tag{decreasing $\mix_{\run}$}
\end{align}%
\usetagform{default}%
\noindent
where we set $K = \vdim\smooth \diam(\points) \max_{\point\in\points} \abs{\pay(\point)}$.
Moreover, by invoking \cref{lem:snoise,lem:sbias,lem:svar}, the term \eqref{eq:reg-bound-lems} becomes
\begin{equation}
\step_{\run} \exof{\snoise_{\head_{\run}+1}\given \filter_{\run}}
		+ \step_{\run} \exof{\sbias_{\head_{\run}} \given \filter_{\run}}
		+ \step _{\run}^{2} \exof{\svar_{\run}^{2} \given \filter_{\run}}
	\leq 0
		+ \bbound \step_{\run} \mix_{\head_{\run}}
		+ \frac{\vbound^{2}}{2} \frac{\step_{\run}^{2}}{\mix_{\head_{\run}}^{2}},
\end{equation}
with $\vbound^{2} = \vdim^{2} \max_{\point\in\points} \abs{\point(\point)}$ as in the proof of \cref{lem:svar}.
Thus, putting everything together,
and recalling that $\pay_{\run}$ is assumed concave,
we get
\begin{align}
\label{eq:temp-bound2}
\pay_{\run}(\point) - \pay_{\run}(\state_{\run})
	\leq \braket{\payv(\state_{\run})}{\point - \state_{\run}}
	&\leq \frac{\norm{\state_{\run} - \point}^{2} - \exof{\norm{\state_{\run+1} - \point}^{2} \given \filter_{\run}}}{2\step_{\run}}
	\notag\\
	&+ \frac{1}{\step_{\run}}
		\bracks*{
			K \step_{\run} \sum_{\runalt=\head_{\run}}^{\run-1} \frac{\step_{\runalt}}{\mix_{\runalt}}
			+ \bbound \step_{\run}\mix_{\head_{\run}}
			+ \frac{\vbound^{2}}{2} \frac{\step_{\run}^{2}}{\mix_{\head_{\run}}^{2}}
		} 
\end{align}

Now, using \cref{lem:3sums} and the fact that $\step_{\run}$ is decreasing,
a summation of the above yields:
\begin{align}
\sum_{\run=1}^{\nRuns} \one_{\{\pool_{\run}\neq\varnothing\}} \bracks{\pay(\point) - \pay(\state_{\run})}
	&\leq \sum_{\run=\start}^{\nRuns} \frac{\norm{\state_{\run} - \point}^{2} - \exof{\norm{\state_{\run+1} - \point}^{2} \given \filter_{\run}}}{2\step_{\run}}
	\notag\\
	&\qquad
		+ \frac{1}{\step_\nRuns}
			\sum_{\run=\start}^{\nRuns}
				\one_{\{\pool_{\run}\neq\varnothing\}}
				\bracks*{K A_{\run} + \bbound B_{\run} + \frac{\vbound^{2}}{2} C_{t}}
	\notag\\
	&= \sum_{\run=\start}^{\nRuns} \frac{\norm{\state_{\run} - \point}^{2} - \exof{\norm{\state_{\run+1} - \point}^{2} \given \filter_{\run}}}{2\step_{\run}}
		+ \tilde\bigoh\parens*{\frac{1}{\step_\nRuns}}
	\notag\\
	&= \frac{\norm{\state_{\start} - \point}^{2}}{2\step_{\start}}
		- \frac{\exof{\norm{\state_{\run+1} - \point}^{2} \given \filter_{\nRuns}}}{2\step_{\nRuns}}
	\notag\\
	&\qquad
		+ \sum_{\run=2}^{\nRuns}
		\bracks*{
		\frac{\norm{\point-\state_{\run}}^{2}}{2\step_{\run}} -\frac{\exof{\norm{\point-\state_{\run}}^{2}\given \filter_{\run-1}}}{2\step_{\run-1}}
		} 
		+ \tilde\bigoh\parens*{\frac{1}{\step_\nRuns}}.
\end{align}
Then, by taking expectations and letting $\bar\breg_{\run} = (1/2) \exof{\norm{\state_{\run} - \point}^{2}}$, we get:
\begin{align}
\exof*{\sum_{\run=\start}^{\nRuns} \one_{\{\pool_{\run}\neq\varnothing\}}
	\bracks{\pay(\point) - \pay(\state_{\run}))}}
	&\leq \frac{\bar\breg_{\start}}{\step_1}
		- \frac{\bar\breg_{\nRuns+1}}{\step_{\nRuns}}
		+ \sum_{\run=2}^{\nRuns} \bracks*{
			\frac{1}{\step_{\run}} -\frac{1}{\step_{\run-1}}} \bar\breg_{\run}
		+ \tilde\bigoh\parens*{\frac{1}{\step_\nRuns}}
	\notag\\
	&\leq \frac{\diam(\points)^{2}}{2\step_{\start}}
		+ \frac{\diam(\points)^{2}}{2}
			\sum_{\run=2}^{\nRuns}
			\bracks*{\frac{1}{\step_{\run}} -\frac{1}{\step_{\run-1}}}
		+ \tilde\bigoh\parens*{\frac{1}{\step_\nRuns}}
	\notag\\
	&= \frac{\diam(\points)^{2}}{2\step_\nRuns}
		+ \tilde\bigoh\parens*{\frac{1}{\step_\nRuns}}
	= \tilde\bigoh\parens*{\frac{1}{\step_\nRuns}}.
\end{align}

Recall now that the number of steps without an update is $o(\nRuns^\alpha)$ by \cref{lem:delayedfeedback,lem:noupdate}.
We thus conclude that:
\begin{equation}
\exof*{\sum_{\run=1}^{\nRuns} \bracks{\pay(\point) -\pay(\state_{\run}))}}
	= \tilde\bigoh\parens*{\frac{1}{\step_\nRuns} + \nRuns^\alpha}
	= \tilde\bigoh\parens*{\nRuns^\alpha + \nRuns^{c}}.
\end{equation}
Using \cref{fig:b_and_c} as a visual aid, we see that the smallest admissible value of $c$ is $\frac{3}{4}$ if $\alpha\leq\frac{1}{4}$, or $\frac{2}{3}+\frac{\alpha}{3}$ otherwise \textendash\ \ie the intersection of the lines $b=2c-1-\alpha$ and $b=1-c$.
In addition, $c\geq\frac{2}{3}+\frac{\alpha}{3}>\alpha$ because $\alpha<1$.
Therefore, by choosing $c=\max(\frac{3}{4},\frac{2}{3}+\frac{\alpha}{3})$  and $b=\min(\frac{1}{4},\frac{1}{3}-\frac{\alpha}{3})$ to minimize the term $\nRuns^{\alpha} + \nRuns^{c}$, we conclude that
\begin{equation}
\overline\reg(\nRuns)
	= \exof*{\sum_{\run=1}^{\nRuns} (\pay(\point) -\pay(\state_{\run}))}
	= \tilde\bigoh(\nRuns^{c})
	= \tilde\bigoh\parens[\big]{\nRuns^{3/4} + \nRuns^{2/3 + \alpha/3}}.
\end{equation}
Consequently, the $\tilde\bigoh(\nRuns^{3/4})$ bandit bound is achieved for $c=\frac{3}{4}$, $b=\frac{1}{4}$ and $\alpha\leq\frac{1}{4}$.
\end{proof}

\section{Convergence to \acl{NE}}
\label{app:Nash}

We now proceed to prove \cref{thm:Nash} on the convergence to \acl{NE}.
For convenience, we restate our result below:

\Nash*

To streamline our presentation, we divide the proof of \cref{thm:Nash} in two parts:
First, we show that the distance between $\state_{\run}$ and the game's (necessarily unique) \acl{NE} $\eq$ admits a well-defined limit with probability $1$;
subsequently, we extract a subsequence of $\state_{\run}$ that converges to $\eq$.
Proving these two results would imply that the limit of the distance of $\state_{\run}$ to equilibrium is necessarily zero (with probability $1$).
We make this approach precise below:

\begin{proposition}
\label{prop:dist-lim}
Suppose that $\game$ satisfies \eqref{eq:DSC} and let $\eq$ be its \textpar{necessarily unique} \acl{NE}.
Then, with assumptions as in \cref{thm:Nash}, the limit $\lim_{\run\to\infty} \norm{\state_{\run} - \eq}$ exists and is finite with probability $1$.
\end{proposition}

\begin{proof}
By \eqref{eq:DSC}, we have
\begin{equation}
\braket{\payv(\point)}{ \point - \sol}
	\leq 0
	\quad
	\text{for all $\point\in\points$},
\end{equation}
with equality holding if and only if $\point = \sol$.
Then, proceeding as in the proof of the no-regret bound of \cref{thm:regret} and rearranging the \acs{RHS} of \eqref{eq:temp-bound2}, we obtain
\begin{align}
\label{eq:qFejer}
\frac{1}{2}\exof{\norm{\state_{\run+1} - \eq}^{2}\given \filter_{\run}}
	&\leq \frac{1}{2} \norm{\state_{\run} - \eq}^{2}
	+\ \bracks*{ K A_{\run} + \bbound B_{\run} + \frac{\vbound^{2}}{2} C_{\run}}
		\cdot \one_{\pool_{\run}\neq\varnothing}
\end{align}
with $K$ and $\bbound$ positive constants, and $A_{\run}$, $B_{\run}$ and $C_{\run}$ defined in \cref{eq:3seqs}.

To proceed, note that
\begin{equation}
\sum_{\runalt=\run}^{\infty}
	\bracks*{ K A_{\runalt} + \bbound B_{\runalt} + \frac{\vbound^{2}}{2} C_{\runalt}}
		\cdot \one_{\{\pool_{\runalt}\neq\varnothing\}}
	\leq \sum_{\runalt=\run}^{\infty}
		\bracks*{ K A_{\runalt} + \bbound B_{\runalt} + \frac{\vbound^{2}}{2} C_{\runalt}}
	< \infty
\end{equation}
by \cref{lem:3sums}.
Thus, if we let
\begin{equation}
S_{\run}
	= \frac{1}{2} \norm{\state_{\run} - \eq}^{2}
		+ \sum_{\runalt=\run}^{\infty}
			\bracks*{ K A_{\runalt} + \bbound B_{\runalt} + \frac{\vbound^{2}}{2} C_{\runalt}}
			\one_{\{\pool_{\runalt}\neq\varnothing\}},
\end{equation}
the bound \eqref{eq:qFejer} yields
\begin{align}
\exof{S_{\run+1} \given \filter_{\run}}
	&= \frac{1}{2} \exof{\norm{\state_{\run+1} - \eq}^{2} \given \filter_{\run}}
		+\sum_{\runalt=\run+1}^{\infty}
			\one_{\{\pool_{\runalt}\neq\varnothing\}}
			\bracks*{ K A_{\runalt} + \bbound B_{\runalt} + \frac{\vbound^{2}}{2} C_{\runalt}}
	\notag\\
	&\leq \frac{1}{2} \norm{\state_{\run} - \eq}^{2}
	+ \bracks*{ K A_{\run} + \bbound B_{\run} + \frac{\vbound^{2}}{2} C_{\run}}
		\cdot \one_{\pool_{\run}\neq\varnothing}
	\notag\\
	&\hphantom{\leq \frac{1}{2} \norm{\state_{\run} - \eq}^{2} \;}
		+\sum_{\runalt=\run+1}^{\infty}
		\one_{\{\pool_{\runalt}\neq\varnothing\}}
		\bracks*{ K A_{\runalt} + \bbound B_{\runalt} + \frac{\vbound^{2}}{2} C_{\runalt}}
	\notag\\
	&= \frac{1}{2} \exof{\norm{\state_{\run} - \eq}^{2} \given \filter_{\run}}
		+\sum_{\runalt=\run}^{\infty}
			\one_{\{\pool_{\runalt}\neq\varnothing\}}
			\bracks*{ K A_{\runalt} + \bbound B_{\runalt} + \frac{\vbound^{2}}{2} C_{\runalt}}
	= S_{\run}
\end{align}
\ie $S_{\run}$ is a supermartingale (relative to $\filter_{\run}$).
Moreover, by taking expecations, we also get
\begin{align}
\exof{S_{\run}}
	= \exof{\exof{S_{\run} \given \filter_{\run-1}}}
	&\leq \exof{S_{\run-1}}
	\leq \dotsc
	\leq \exof{S_{\start}}
	\notag\\
	&= \frac{1}{2} \norm{\state_{\start} - \eq}^{2}
		+ \exof*{
			\sum_{\runalt=\start}^{\infty}
			\bracks*{ K A_{\runalt} + \bbound B_{\runalt} + \frac{\vbound^{2}}{2} C_{\runalt}}
			\one_{\{\pool_{\runalt}\neq\varnothing\}}}
	< \infty
\end{align}
again by \cref{lem:3sums}.
This shows that $S_{\run}$ is (uniformly) bounded in $L^{1}$ so, by Doob's submartingale convergence theorem \citep[Theorem 2.5]{HH80}, we conclude that $S_{\run}$ converges \as to some finite random variable $S_{\infty}$.
In turn, this implies that $\lim_{\run\to\infty} \norm{\state_{\run} - \eq}$ exists and is finite \as, as claimed.
\end{proof}

Our second result (which is of independent interest) concerns the extraction of a subsequence of $\state_{\run}$ converging to $\eq$.

\begin{proposition}
\label{prop:subsequence}
With assumptions as in \cref{thm:Nash}, there exists with probability $1$ a \textpar{possibly random} subsequence $\state_{\run_{k}}$ of $\state_{\run}$ such that $\lim_{k\to\infty} \state_{\run_{k}} = \eq$.
\end{proposition}

\begin{proof}
Suppose ad absurdum that the event
\begin{equation}
\samples_{0}
	= \{\liminf\nolimits_{\run\to\infty} \norm{\state_{\run} - \eq} > 0\}
\end{equation}
occurs with positive probability (\ie with positive probability, $\state_{\run}$ does not admit $\eq$ as a limit point).
Conditioning on this event, there exists a (nonempty) compact set $\cvx \subset \points$ such that $\eq\notin\cvx$ and $\state_{\run} \in \cvx$ for all sufficiently large $\run$.
Therefore, by the continuity of $\payv$ and the fact that the game satisfies \eqref{eq:DSC}, there exists some $c>0$ such that
\begin{equation}
\label{eq:cbound}
\braket{\payv(\point)}{\point - \sol}
	\leq -c
	< 0
	\quad
	\text{for all $\point\in\cvx$}.
\end{equation}

To proceed, if we telescope \eqref{eq:start} for $\base = \eq$ and we use the decomposition \eqref{eq:signal} of $\est\payv_{\run}$, we get
\begin{align}
\frac{1}{2}\norm{\state_{\run+1} - \eq}^{2}
	&\leq \frac{1}{2}\norm{\state_{1} - \eq}^{2}
	\notag\\
	&+ \sum_{\runalt=\start}^{\run}
		\step_{\runalt}
		\one_{\{\pool_{\runalt}\neq\varnothing\}}
		\braket{\payv(\state_{\runalt})}{\state_{\runalt} - \eq}
	\notag\\
	&+ \sum_{\runalt=\start}^{\run} 
		\step_{\runalt}
		\one_{\{\pool_{\runalt}\neq\varnothing\}}
		\sum_{k=\head_{\runalt}}^{\runalt-1} \braket{\payv(\state_{k}) - \payv(\state_{k+1})}{\state_{k} - \eq}
	\notag\\
	&+ \sum_{\runalt=\start}^{\run}
		\step_{\runalt}
		\one_{\{\pool_{\runalt}\neq\varnothing\}}
		\snoise_{\head_{\runalt}+1}
	+ \sum_{\runalt=1}^{\run}
		\step_{\runalt}
		\one_{\{\pool_{\runalt}\neq\varnothing\}}
		\sbias_{\head_{\runalt}}
	+ \sum_{\runalt=\start}^{\run} \one_{\{\pool_{\runalt}\neq\varnothing\}} \step_{\runalt}^{2} \svar_{\runalt}^{2}
\end{align}
with $\snoise_{\head_{\run}+1}$, $\sbias_{\head_{\run}}$ and $\svar_{\run}$ defined respectively as in \cref{lem:snoise,lem:sbias,lem:svar}.
Subsequently, letting $\nu_{\run} =\sum_{\runalt=\start}^{\run} \step_{\runalt}$ and
\begin{equation}
Q_{\run}
	= \sum_{\runalt=\start}^{\run} 
		\step_{\runalt}
		\one_{\{\pool_{\runalt}\neq\varnothing\}}
		\sum_{k=\head_{\runalt}}^{\runalt-1} \braket{\payv(\state_{k}) - \payv(\state_{k+1})}{\state_{k} - \eq},
\end{equation}
the bound \eqref{eq:cbound} yields
\begin{align}
\label{eq:subseq-bound}
\frac{1}{2}\norm{\state_{\run+1} - \eq}^{2}
	&\leq \frac{1}{2} \norm{\state_{\start} - \eq}^{2}
		- c\nu_{\run}
	\notag\\
	&+ \bracks*{
		Q_{\run}
		+ \sum_{\runalt=\start}^{\run} \step_{\runalt} \one_{\{\pool_{\runalt}\neq\varnothing\}}\snoise_{\head_{\runalt}+1}
		+ \sum_{\runalt=1}^{\run} \step_{\runalt} \one_{\{\pool_{\runalt}\neq\varnothing\}} \sbias_{\head_{\runalt}}
		+ \sum_{\runalt=\start}^{\run} \one_{\{\pool_{\runalt}\neq\varnothing\}} \step_{\runalt}^{2} \svar_{\runalt}^{2}
		} 
\end{align}
We now proceed term-by-term:
\begin{enumerate}
\item
For the term $Q_{\run}$, working as in the case of \eqref{eq:reg-bound-unfold} in the proof of \cref{thm:regret}, we have:
\begin{align}
Q_{\run}
	&= \sum_{\runalt=\start}^{\run}
		\step_{\runalt}
		\sum_{k=\head_{\runalt}}^{\runalt-1} \braket{\payv(\state_{k}) - \payv(\state_{k+1})}{\state_{k} - \eq}
	\notag\\
	&\leq \sum_{\runalt=\start}^{\run} 
		\step_{\runalt}
		\cdot K
			\sum_{k=\head_{\runalt}}^{\runalt-1} \frac{\step_{\runalt}}{\mix_{\runalt}}
	\leq K \sum_{\runalt=\start}^{\infty} A_{\runalt}
	< \infty,
\end{align}
where $A_{\run}$ is defined as in \eqref{eq:3seqs} and, in the last step, we used \cref{lem:3sums}.

\item
For the term involving $\snoise$, let
\begin{equation}
M_{\run}
	= \sum_{\runalt=\start}^{\run} \step_{\runalt} \one_{\{\pool_{\runalt}\neq\varnothing\}}\snoise_{\head_{\runalt}+1}.
\end{equation}
Since $\pool_{\run}$ is $\filter_{\run}$-measurable, we have
\begin{equation}
\exof{M_{\run} \given \filter_{\run}}
	= M_{\run-1} + \step_{\run} \one_{\{\pool_{\run}\neq\varnothing\}} \exof{\snoise_{\head_{\run+1}} \given \filter_{\run}}
	= M_{\run-1},
\end{equation}
\ie $M_{\run}$ is a martingale.
Moreover, by the construction of the \ac{SPSA} estimator \eqref{eq:oracle} and the decomposition \eqref{eq:signal} thereof, there exists some constant $\sigma >0$ such that
\begin{equation}
\norm{\noise_{\run+1}}^{2}\leq\frac{\sigma^{2}}{\mix_{\run}^{2}}.
\end{equation}
\normalsize
We thus get:
\begin{align}
\exof{M_{\run}^{2}}
	&\leq \sum_{\runalt=\start}^{\infty}
		\one_{\pool_{\runalt}\neq\varnothing}
		\step_{\runalt}^{2}
		\exof{\snoise^{2}_{\head_{\runalt}+1} \given \filter_{\runalt}}
	\notag\\
	&\leq \sum_{\runalt=1}^{\infty}
		\one_{\pool_{\runalt}\neq\varnothing}
		\step_{\runalt}^{2}
		\norm{\state_{\runalt}-\sol}^{2}
		\exof{\norm{\noise_{\head_{\runalt+1}}}^{2} \given \filter_{\runalt}}
	\notag\\
	&\leq \sum_{\runalt=1}^{\infty}
		\step_{\runalt}^{2}
		\norm{\state_{\runalt}-\sol}^{2}
		\exof{\norm{\noise_{\head_{\runalt+1}}}^{2} \given \filter_{\runalt}}
	\notag\\
	&\leq \diam(\points)^{2} \sigma^{2}
		\sum_{\runalt=1}^{\infty} \frac{\step_{\runalt}^{2}}{\mix_{\runalt}^{2}}
	<\infty
\end{align}
Therefore, by the law of large numbers for martingale difference sequences \citep[Theorem 2.18]{HH80}, we conclude that $M_{\run} / \nu_{\run}$ converges to $0$ with probability $1$ \textendash\ and hence, with probability $1$ conditioned on $\samples_{0}$.

\item
For the second-to-last term in the brackets of \eqref{eq:subseq-bound}, \cref{lem:sbias} readily yields:
\begin{equation}
\sum_{\runalt=1}^{\run} \step_{\runalt} \one_{\{\pool_{\runalt}\neq\varnothing\}} \sbias_{\head_{\runalt}}
	\leq \sum_{\runalt=\start}^{\infty} \step_{\runalt} \sbias_{\head_{\runalt}}
	\leq \bbound \sum_{\runalt=\start}^{\infty} B_{s}
	< \infty,
\end{equation}
where $B_{\run}$ is defined as in \eqref{eq:3seqs} and, in the last step, we used \cref{lem:3sums}.

\item
Finaly, for the last term, we have:
\begin{equation}
\sum_{\runalt=\start}^{\run} \one_{\{\pool_{\runalt}\neq\varnothing\}} \step_{\runalt}^{2} \svar_{\runalt}^{2}
	\leq \sum_{\runalt=\start}^{\infty} \step_{\runalt}^{2} \svar_{\runalt}^{2}
	\leq \frac{\vbound^{2}}{2} \sum_{\runalt=\start}^{\infty} \frac{\step_{\runalt}^{2}}{\mix_{\head_{\runalt}}^{2}}
\end{equation}
where the last step follows from \cref{lem:svar}.
Subsequently, by \cref{lem:3sums}, we have $\sum_{\runalt=\start}^{\infty} \step_{\runalt}^{2} / \mix_{\head_{\runalt}}^{2} < \infty$, so this term is also finite.
\end{enumerate}

From the above, we conclude that, with probability $1$ on $\samples_{0}$ (and hence, with positive probability overall), all the non-constant terms in the brackets of \eqref{eq:subseq-bound} converge to $0$.
In turn, this implies that
\begin{equation}
\norm{\state_{\run+1} - \eq}^{2}
	\leq \norm{\state_{1} - \eq}^{2}
	- \abs{\Theta(\nu_{\run})}
	\to -\infty
	\quad
	\text{as $\run\to\infty$},
\end{equation}
a contradiction.
Going back to our original assumption, this shows that, with probability $1$, $\state_{\run}$ admits $\eq$ as a limit point, and our proof is complete.
\end{proof}


Putting all this together, the proof of \cref{thm:Nash} is relatively straightforward:

\begin{proof}[Proof of \cref{thm:Nash}]
Under the stated assumptions, $\norm{\state_{\run} - \eq}$ converges to some finite value with probability $1$.
Given that $\state_{\run}$ admits a subsequence converging to $\eq$ \as, we conclude that $\lim_{\run\to\infty} \norm{\state_{\run} - \eq} = 0$, again with probability $1$.
Finally, since $\mix_{\run}$ is decreasing to $0$, we have
\begin{equation}
\norm{\est\state_{\run} - \eq}^{2}
	= \norm{\state_{\run} + \mix_{\run}\pertvar_{\run} - \eq}^{2}
	\leq 2 \norm{\state_{\run} - \eq}^{2} + 2 \mix_{\run}^{2} \norm{\pertvar}^{2}
	\to 0
	\quad
	\as,
\end{equation}
\ie the sequence of generated actions $\est\state_{\run}$ converges to the game's \acl{NE} with probability $1$, as claimed.
\end{proof}

\bibliographystyle{plainnat}
\bibliography{bibtex/IEEEabrv,bibtex/Bibliography-GOLD}

\end{document}